\documentclass[ aps,
pre,
twocolumn,              
10pt,                   
superscriptaddress,     
preprintnumbers,        
nofootinbib,             
showpacs,               
showkeys,               
]{revtex4-2}

\usepackage{graphicx}               
\usepackage[caption=false]{subfig}  
\usepackage[svgnames]{xcolor}       
\usepackage{xspace}                 

\usepackage[T1]{fontenc}                 
\usepackage{lmodern}                     
\usepackage[scaled=.8]{beramono}         
\usepackage{microtype}                   

\usepackage{dcolumn}                    
\usepackage{booktabs}                   
\usepackage{longtable}                  

\usepackage{old-arrows}                 
\usepackage{url}                        
\usepackage[super]{nth}                 
\usepackage{lipsum}                     
\usepackage{siunitx}                    
\usepackage{algpseudocodex}             
\usepackage{csquotes}                   

\usepackage{amsmath}        
\usepackage{amsfonts}       
\usepackage{amssymb}        
\usepackage{amsthm}         

\theoremstyle{plain}            
\theoremstyle{plain}            \newtheorem{corollary}{Corollary}
\theoremstyle{plain}            \newtheorem{theorem}{Theorem}
\theoremstyle{plain}            \newtheorem{proposition}{Proposition}
\theoremstyle{plain}            
\theoremstyle{remark}           
\theoremstyle{definition}       \newtheorem{definition}{Definition}

\usepackage{overarrows}                     
\usepackage{mathtools}                      
\usepackage{mathrsfs}                       
\usepackage{upgreek}                        
\usepackage{braket}                         
\usepackage{stmaryrd}                       
\SetSymbolFont{stmry}{bold}{U}{stmry}{m}{n} 
\usepackage{nicefrac}                       
\usepackage{bm}                             
\usepackage{bbm}                            

\usepackage{hyperref}           
\hypersetup{
    colorlinks   = true,        
    linkcolor={blue!50!black},
    citecolor={blue!50!black},
    urlcolor={blue!80!black}
    }

\usepackage[capitalise]{cleveref}       
\newcommand{\eM}     {\mbox{$\epsilon$-machine}\xspace }
\newcommand{\eMs}    {\mbox{$\epsilon$-machines}\xspace }
 
\newcommand{\EMs}    {\mbox{$\epsilon$-Machines}\xspace }
\newcommand{\eT}     {\mbox{$\epsilon$-transducer}\xspace }
\newcommand{\eTs}    {\mbox{$\epsilon$-transducers}\xspace }

\NewOverArrowCommand{\PastArrow}{start=\leftharpoonup,
                                 end=\relbar,
                                 trim=6,
                                 space after arrow=-.1ex}

\NewOverArrowCommand{\pastArrow}{start={\smallermathstyle\leftharpoonup},
                                 end=\relbar,
                                 trim start = 5,
                                 trim end = 8,
                                 }
                
\NewOverArrowCommand{\FutureArrow}{end=\rightharpoonup, 
                                   trim=7,
                                   space after arrow=-.1ex}
    
\NewOverArrowCommand{\futureArrow}{start={\smallermathstyle\relbar}, 
                                   end=\rightharpoonup,
                                   trim start = 5,
                                   trim end = 8,
                                  }
                                
\NewOverArrowCommand{\BiArrow}{start=\leftharpoonup,
                                 end=\rightharpoonup,
                                 trim=8,
                                 space after arrow=-.1ex}

\NewOverArrowCommand{\biArrow}{start={\smallermathstyle\leftharpoonup},
                                 end=\rightharpoonup,
                                 trim=9
                                 }

\newcommand{\Machine}               { M }
\newcommand{\eMachine}              { \Machine_{\epsilon} }

\newcommand{\EquiFunction}[1]       { \epsilon \left[ #1\right] }
\newcommand{\Process}               { \mathcal{P} }

\newcommand{\MeasSymbol}            { {X} }
\newcommand{\meassymbol}            { {x} }
\DeclareRobustCommand{\BiInfinity}            { \BiArrow {\MeasSymbol} }

\DeclareRobustCommand{\Past}        { \PastArrow{\MeasSymbol} }
\newcommand{\PastSmashed}           { \smash{\Past} }

\newcommand{\past}                  { \pastArrow{\meassymbol} }
\newcommand{\pastprime}             { \past^{\prime} }

\DeclareRobustCommand{\Future}      { \FutureArrow {\MeasSymbol} }
\newcommand{\FutureSmashed}         { \smash{\Future} }

\newcommand{\future}                { \futureArrow {\meassymbol} }

\newcommand{\CausalState}           { \mathcal{S} }
\newcommand{\CausalStatePrime}      { \CausalState^{\prime}}
\newcommand{\CausalStateSet}        { \bm{\CausalState} }

\newcommand{\causalstate}           { \sigma }

\newcommand{\CausalStateTransition} { {\CausalState \to \CausalStatePrime} }

\newcommand{\LabelCausalTransition} { T^{(\msym)}_{\CausalStateTransition} }
\newcommand{\CausalTransitionSet}   { \{ \LabelCausalTransition:
\msym \in \alphabet \} }

\newcommand{\selfi}{\mathfrak{i}}
\newcommand{\SelfI}{\mathfrak{I}}

\newcommand{\selfii}[1]                { \selfi \left[ #1 \right] }
\newcommand{\SelfII}[1]                { \SelfI \left[ #1 \right] }
\newcommand{\selficond}[2]             { \selfi \left[ #1 \mid #2 \right] }  
\newcommand{\SelfIcond}[2]             { \SelfI \left[ #1 \mid #2 \right] }  
\newcommand{\selfijoint}[2]            { \selfi \left[ #1 , #2 \right] }  
\newcommand{\SelfIjoint}[2]            { \SelfI \left[ #1 , #2 \right] }  
\newcommand{\selfimut}[2]              { \selfi \left[ #1 : #2 \right] }  
\newcommand{\SelfImut}[2]              { \SelfI \left[ #1 : #2 \right] }  

\newcommand{\I}[2]                  { \operatorname{I} \left[ #1 : #2 \right] }
\newcommand{\Icond}[3]              { \operatorname{I} \left[ #1 : #2 \mid #3 \right] } 
\newcommand{\ExcessEntropy}         { \bm{E} }
\newcommand{\EE}                    { \ExcessEntropy }

\newcommand{\Cmu}                   { C_{\mu} }

\newcommand{\hmu}                   { h_{\mu} }
\newcommand{\rhomu}                 { \rho_{\mu} }
\newcommand{\rmu}                   { r_{\mu} }
\newcommand{\qmu}                   { q_{\mu} }
\newcommand{\bmu}                   { b_{\mu} }
\newcommand{\sigmu}                 { \sigma_{\mu} }
\newcommand{\bmuforward}            { \bmu^+ }
\newcommand{\bmureverse}            { \bmu^- }

\newcommand{\integers}              { \mathbb{Z} }

\newcommand{\sigmaAlgebra}          { \Sigma }

\newcommand{\measure}               { \mu }

 \newcommand{\RVSet}
{ \mathfrak{X} }

\newcommand{\shiftOperator}     { \uptau }
\newcommand{\alphabet}          { \mathcal{\MeasSymbol} }

\newcommand{\numSyms}           { k }

\newcommand{\Present}           { \MeasSymbol_0 }

\newcommand{\MSym}{\MeasSymbol}
\newcommand{\msym}{\meassymbol}
\newcommand{\MS} [2] {\ensuremath{\MeasSymbol_{#1:#2}}\xspace} \newcommand{\ms}
[2] {\ensuremath{\meassymbol_{#1:#2}}\xspace}

\newcommand{\MeasSymbolVarTwo}{Z}

\begin{document}

\title{Agentic Information Theory:\\
Ergodicity and Intrinsic Semantics of\\
Information Processes}

\author{James P. Crutchfield}
\email{chaos@ucdavis.edu}
\homepage{http://csc.ucdavis.edu/~chaos/}
\affiliation{Complexity Sciences Center and Physics Department \\ 
             University of California at Davis, One Shields Avenue, Davis, CA
             95616}

\author{Alexandra Jurgens}
\email{alexandra.jurgens@inria.fr}
\homepage{http://csc.ucdavis.edu/~ajurgens/}
\affiliation{GEOSTAT Team, INRIA -- Bordeaux Sud Ouest \\
             33405 Talence Cedex, France}


\begin{abstract} 
We develop information theory for the temporal behavior of memoryful agents
moving through complex---structured, stochastic---environments. We introduce
and explore information processes---stochastic processes produced by cognitive
agents in real-time as they interact with and interpret incoming stimuli. We
provide basic results on the ergodicity and semantics of the resulting time
series of Shannon information measures that monitor an agent's adapting view of
uncertainty and structural correlation in its environment.
\end{abstract}


\date{\today}

\preprint{arxiv.org:2505.19275 [cond-mat.stat-mech]}


\maketitle
\tableofcontents


\section{Introduction}
\label{sec:introduction}

Contemporary statistical mechanics \cite{Cros09a,Stre09a,Seif12a} has made
considerable contributions to machine learning
\cite{Hert86a,Mack03a,Sohl15a,Bahi20a}, from formal models and algorithms
(e.g., reinforcement learning \cite{Sutt18a,Rahm19a}) to implementing
thermodynamic systems for learning tasks \cite{Mela23a,Cole23a}. This
cross-fertilization reflects broader recent concerns, though, about emergent
organization in active matter \cite{Bech16a}, collective intelligence
\cite{Suli20a}, embodiments of biological intelligence \cite{Gupt21a}, and
distributed artificial intelligence. Whether groups of molecules, animals, or
robots, these systems consist of many interacting components---often referred
to generically as ``agents'' in the complex adaptive systems and machine
learning literatures \cite{Bona99a,Simo96a,Holl75a,Lang89a,Sutt18a}. These
agents glean information from their environment and use it to take actions,
that in turn, affect the environment. One notable phenomenon is that these
collections can organize into temporal and spatial behaviors beyond those of
the individual agents. Today, this kind of emergent collective organization is
becoming a concern as machine learning moves ever closer to massive deployments
of adaptive agents on globe-spanning communication networks.

Characteristically, though, statistical mechanical approaches to collective
thermodynamic functioning and information processing focus on asymptotic
quantities---such as, say, correlation functions, time- or ensemble-average
rates of work extraction, dissipation, and information storage. As such, they
give only the barest insight into the actual, online, and time-dependent
operation of an agent's internal mechanisms---the underlying information heat
engines \cite{Maru09a,Mand12a,Parr15a,Saga12a,Boyd15a} that support relevant
behavior and functioning.

The following addresses informational aspects of such real-time adaptive
agents, ultimately referring to the \emph{information processes} that reflect
how a stochastic dynamical system acquires, stores, and internally processes
information. In effect, these real-time adaptive agents directly manipulate
information available in observations. This view begins to move away from
indirect appeals to frequentist or Bayesian probability for defining and
estimating information, instead to directly operating on inputs that are innately informative. Fundamental to this is prediction.
This informational approach, though, entails several computational challenges.
Fortunately, Ref. \cite{Jurg25a} gives explicit and efficient methods for
calculating the underlying informational measures from a stochastic process'
optimal representation---the \eM. And so, the conceit in the following is that
a cognitive agent implements computational mechanics to learn internal models
of its environment.

\subsection{Challenge}
\label{sec:Challenge}

The following lays out the basic concept of information processes---stochastic
processes that capture what a cognitive agent observes and
measures---informationally---from a stochastic process generated by an
environment. It reviews prerequisites from stochastic processes, ergodic
theory, and information theory.

To emphasize, such agents behave as more than memoryless particles or
memoryless input-output mappings. They process information. In this literal
sense the agents are \emph{cognitive}---as they observe their environment they,
explicitly or implicitly, make various interpretations of stimuli---apparent
correlation and uncertainty and available free energy, for example. The latter
are expressed as kinds of information estimated in real-time as an agent adapts
to the time series of incoming stimuli and behaves accordingly.

To understand the agent's moment-by-moment operation we track various quantities
estimated during the agent's interactions with its environment. Specifically, we
monitor the interactions via Shannon's information atoms \cite{Yeun08a}. In this
setting, then, an agent is a transducer that maps from a given process (the
environment) to one or several observational stochastic processes---information
processes comprised of time series of instantaneous entropy rate (apparent
randomness), bound information (apparent information storage), elusive
information (apparent hidden information), and so on.

We show that if the environment being observed is stationary then the resulting
information processes developed by the agent are stationary, under certain
conditions. These conditions include observing the environment through a
sliding-window of finite duration or through a function---a statistic---with
limited-range memory. If so, then the agent's interpretations of the
environment's behavior are statistically well-behaved. We also address
subtleties of nonstationarity that arise at early times through transients
during agent-environment synchronization \cite{Jame10a} and through the agent's
employing incorrect internal models. We give companion results on the
ergodicity of such information processes. We then turn to discuss how
information processes convey meaning to an agent with an internal model
(correct or misleading) of its environment.

\subsection{Overview}
\label{sec:Overview}

The immediately following section provides background, by drawing parallels and
even motivations from prior results on complex adaptive systems. It then reviews
elementary information theory and Shannon information measures. These are then
adapted to the online temporal setting of a cognitive agent interacting with
its environment---the stochastic process it observes. Cognitive agents are
introduced using the structural theory of information in stochastic
processes---computational mechanics
\cite{Shal98a}.

In short, a cognitive agent builds a minimal optimal predictor---an \eM---of
its environment. A key point is that agent behavior and its interpreting the
environment are themselves stochastic processes. Taken altogether, this then
gives the foundation for how an agent attributes meaning to the results of
interacting with its environment. In short, there are two levels of semantics,
one focused on the agent's uncertainties in predicting the environment, the
other focused on the agent's informational interpretations of correlation and
structure.

In this way, we develop methods and quantities to monitor how an agent
interprets its interactions with a structured, stochastic environment. Here, we
assume that the environment is stationary---its structure and stochasticity are
time independent---and that the agent begins interacting by employing a given
fixed model of the environment's stochasticity and structure---its internal
model is time-independent. The agent only observes---it takes no actions on the
environment. For now, these are the main simplifications. We then outline how
the agent can use its internal model to interpret the environment's behavior
moment-by-moment. One goal is to show that the resulting information processes
are statistically well-behaved and so can form the basis for further---say,
downstream---estimation and interpretation. Sequels relax these assumptions.

\begin{figure}
\includegraphics[width=\columnwidth]{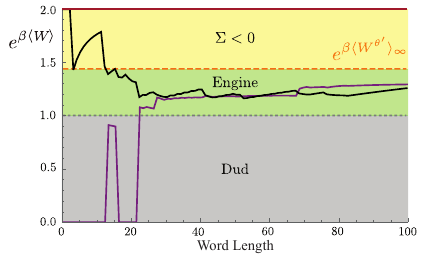}
\caption{Monitoring online thermodynamic performance via the work rate
	$\exp{\beta \langle W^{\theta^\mathrm{max}_\ell} (y_{0:\ell})\rangle}/\ell$
	during learning:
	The information source is a length $\ell = 100$ sample output of a
	$5$-state hidden Markov model (\eM); extracted from Fig. 6 in Ref.
	\cite{Boyd24a} (used with permission). The work rate is estimated using $2$
	memory states. The resulting work-rate monitoring processes are complicated
	and dynamic. Indeed, the engine jumps between different functionings as a
	``Dud'' (gray)---unable to harvest work; an ``Engine'' (green)---doing so
	successfully; and violating the Second Law (yellow)---negative entropy
	fluctuation. The gray dashed line at $1.0$ corresponds to zero work
	production: $\beta\langle{W^\theta}\rangle_\infty = 0$. The orange dashed
	line for work production corresponds to correctly estimating the true model
	$\beta \langle W^{\theta^\prime}\rangle_\infty$.  The solid red line at
	$\beta\langle{W^\theta}\rangle_\infty = \ln 2$ is the maximum work that can
	be harvested per bit from a single sequence.
	}
\label{fig:ThermoSignal}
\end{figure}

\section{Background}
\label{sec:Background}

Before delving into technicalities it will be helpful later on to have in mind
a range of example agent-environment systems---systems in which (i)
extracting information from incoming signals and (ii) interpreting that
information play a key role in agent behavior and collective functioning.

\subsection{Collective Information Processes}
\label{sec:Collectives}

A paradigmatic agent-environment system is that observed in a collective of
flocking animals \cite{Visc95a}. The social behaviors of fish shoaling and
schooling and starling murmurations come immediately to mind
\cite{Shaw78a,Couz05a}. The animals are the agents that take in information
from the spatial surroundings and, importantly, the nearby animals. From this
they determine their motion. The result can be the emergence of stunningly,
captivating complex coordinated spatio-temporal patterns.

One of the puzzles is how this complexity emerges from the given local, but
distributed equations of motion. Using time-dependent, time-delayed mutual
information Ref. \cite{Suli20a} showed how a leader-follower relationship
between animals emerges, eventually leading to a single individual guiding the
flock's collective motion. In this, we see a real-world example of agents (i)
interacting via an environment that (ii) effectively use one or several
information processes to control local behaviors from which collective
organization emerges.

\subsection{Information Embedded in Thermodynamic Systems}
\label{sec:PhysicalProcesses}

While a primary concern here, information is not the only narrative framing for
cognitive agents. For example, analogous \emph{thermodynamic processes}
describe an agent's creation and use of physical resources. Of late, Maxwell's
demon is the oft-quoted example \cite{Szil29a}. In this, agents are
viewed as information heat engines. That is, there are other closely-related
stochastic processes of interest. In particular, if we are interested in an
agent's adaptive physical functioning, then a number of thermodynamic-resource
processes capture engine performance.

The processes generated by Ref. \cite{Geir23a}'s quantum machines and Ref.
\cite{Boyd16c} self-correcting information engines come to mind. To be more
concrete about such processes, however, consider a system that learns and then
validates an information heat engine through environment interactions. Figure
\ref{fig:ThermoSignal}, excerpted from Ref. \cite{Boyd24a}, shows the time
evolution of an information engine's work rate $W^{\theta^\mathrm{max}}$. An
exponential average work rate is plotted for both learning (purple line) and
validation (black line) versus time series length $\ell$. The point is that
such agentic signals track the agent's evolution during learning and its
functioning across time. They are stochastic processes---functions of the
engine's internal, on-going adaptation and operation. And, they are signals
in the design and application of information engines that are essential to monitor.

Relying on such signals, though, to monitor learning or online performance, one
needs to know whether or not they are statistically well-behaved
processes---signals that can be used to quantitatively (convergently)
characterize learning and thermodynamic efficiency. Are they stationary and
ergodic?

Reference \cite{Crut16a} illustrates a higher-level of stochastic variation in
an agent's very thermodynamic functioning. Is it acting as an engine extracting
work from a heat reservoir or as an information eraser that correlates
information? The several examples here illustrate that there are many
circumstances when time-local instantaneous statistics and their operational
meaning are key to monitoring adaptive physical behaviors.

\begin{figure}
\includegraphics[width=\columnwidth]{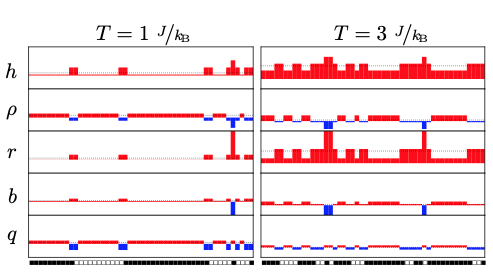}
\caption{Information processes in spin configurations: Motif entropy-component
	analysis of the 1D Ising model at two temperatures $T$, where $\hmu$ is
	the local entropy density, $\rho_\mu$ the anticipated information, $\rmu$
	the ephemeral information, $\bmu$ the bound information, and $\qmu$ the
	enigmatic information. ($J$ is the nearest-neighbor spin coupling strength
	and $k_\text{B}$ is Boltzmann's constant.) A segment of a spin
	configuration with up spins (white cells) and down spins (black cells)
	shown in the bottom row. (Figure 6, Ref. \cite{Vija15a}, used with
	permission.)
	}
\label{fig:SpinInfos}
\end{figure}

\subsection{Information Processes in Biology}
\label{sec:Biological}

Contrasted to physical systems for which information often simply stands in for
probabilistic properties, information plays an essential role in the
functioning of life processes \cite{Schr44a}. Examples abound. For
example, Ref. \cite{Marz18a} argues that bacteria are environmental prediction
engines---their reproduction and survival require them to take actions based
the information that they extract and store from their surroundings, moment by
moment. More to the point, it identifies \emph{instantaneous predictive
information} as a key signal---information shared between an organism's present
phenotype and future environment states. And, it demonstrates that optimal
epigenetic markers are minimal sufficient statistics for evolutionary
prediction---the causal states introduced shortly. Bacteria in this setting are
agents and their surroundings their environment.

Reference \cite{Marz17c} develops a similar analysis of molecular sensors. In
particular, it calls out the role of information gleaned by molecular sensors
viewed as Markovian communication channels embedded in biological cell
membranes.

\subsection{Organization in Statistical Mechanics}
\label{sec:StatMech}

The spatial organization that spontaneously emerges in one-dimensional and
two-dimensional spin systems, though arising in mere physical systems, presents
a similar challenge \cite{Vija15a}. Figure \ref{fig:SpinInfos}, from there,
shows a suite of information processes as a function of spin location in
lattice configurations. Specifically, it plots the local entropy density
$\hmu$, anticipated information $\rho_\mu$, ephemeral information $\rmu$, bound
information $\bmu$, and enigmatic information $\qmu$ across lattice sites.
Reference \cite{Vija15a} then goes on to interpret what these quantities mean
in terms of the underlying physical interactions and emergent spin patterns. We
return to these information signals later on, as they play a key role.

\subsection{Approach}
\label{sec:Approach}

Looking across these very different complex adaptive systems reveals common
features: (i) framing the overall system in terms of one or a group of agents
and their interactions with an environment; (ii) the agent's role in monitoring
the environment by extracting real-time signals---sequences of various kinds
of informational measure; and (iii) the semantics in those measures---how an
agent uses or interprets such signals. Figure \ref{fig:MSemantics} illustrates
the basic measurement channel picture of the environment and agent---a picture
that frames our overall development.

\begin{figure}
\includegraphics[width=\columnwidth]{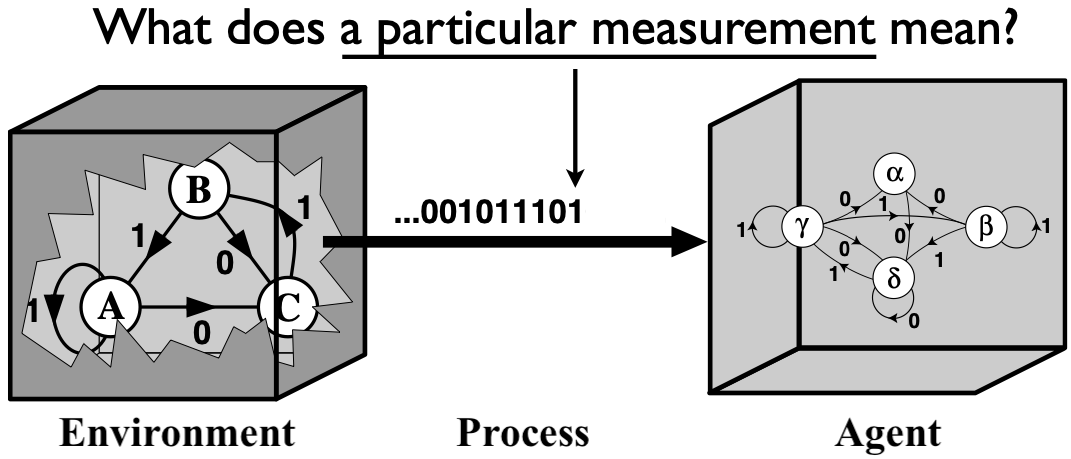}
\caption{Measurement semantics: The channel consists of the environment
	(left) that generates a process (middle)---the result of a series of
	observations and interpretations made by an agent (right).
	}
\label{fig:MSemantics}
\end{figure}

Historically, calculating these informational signals has been quite
theoretically and computationally demanding. Very recently, though, Ref.
\cite{Jurg25a} showed how to perform these calculations using practical and
efficient algorithms, progress that stimulates much of the following. Moreover,
the companion Ref. \cite{Crut24a} introduces out a thorough-going
information-theoretic framing for deploying the information measures we
develop in the following.

With this background and these motivations laid out, we are now ready to turn
to our theoretical development.

\section{Stochastic Processes}
\label{sec:Processes}

The following briefly summarizes discrete stochastic processes, adapting the
more general framework of Ref. \cite{Loom21b} to this simpler setting.
(For measure-theoretic background see Refs. \cite{Bill95a,Gray09a}.)

\newcommand{\TIndices}{\mathbb{Z}}

\renewcommand{\Process}{\MeasSymbol}

A (discrete-time, discrete-valued) \emph{stochastic process} $\Process$
consists of a sequence of indexed random variables (RVs) $\{ \MeasSymbol_t
\}_{t \in \mathbb{Z}}$ defined on a probability space $\left(
\alphabet^{\integers}, \sigmaAlgebra, \measure \right)$. (See Fig.
\ref{fig:DiscreteProcess}.) $\MeasSymbol_t$ is the random variable representing
a value $x \in \alphabet$ observed at time $t$.  $\alphabet$ is the RVs' common
state space or event space---the set of possible observed values or
$\Process$'s \emph{alphabet}. We take it to be a finite set $\alphabet = \{ 1,
\ldots , k\}$.

A process \emph{realization} $\omega$ is a single outcome of the stochastic
process. That is, realizations are taken from the \emph{sample space}: $\omega
\in\alphabet^{\integers}$. This is the collection of a process' behaviors.
$\sigmaAlgebra$ is the sigma algebra that describes the process'
events---possible subsets of realizations.

The measure $\mu$ is defined on measurable subsets $A$ of realizations:
\begin{align}
\mu(A) = \int_{\omega \in A} d \mu(w) ~.
\label{eq:sets}
\end{align}
for $A \in \sigmaAlgebra$; that is, measurable subsets of
$\MeasSymbol^\mathbb{Z}$. As we detail shortly, the measure over bi-infinite
realizations determines the probability of sets of realizations and of finite
sequences via integration over subsets, as in Eq. \eqref{eq:sets}.


\begin{figure}
\includegraphics[width=\columnwidth]{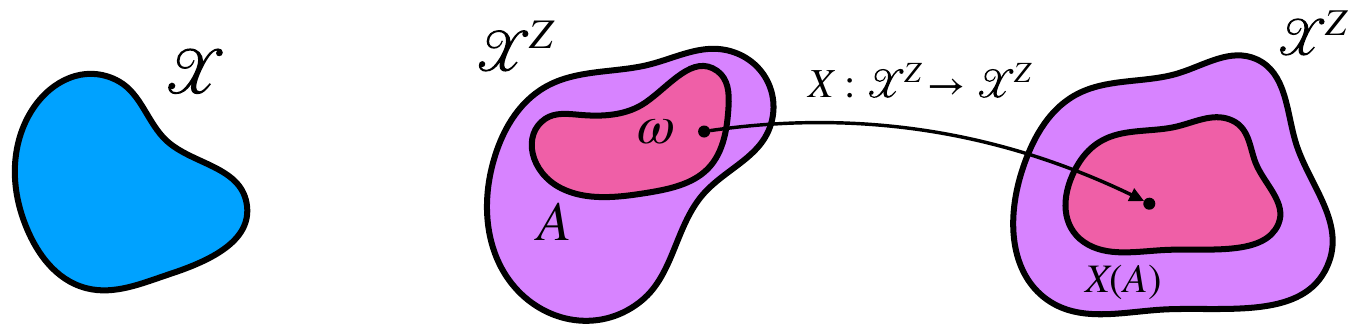}
\caption{Discrete-time, discrete-value stochastic process as a sequence
$\alphabet^\tau$ of indexed random variables $\{\MeasSymbol_t:
t \in \mathbb{Z}\}$ on the probability space $\left( \alphabet^\mathbb{Z},
\Sigma, \mu\right)$ over event space $\alphabet$ with realizations $\omega \in
\alphabet^\mathbb{Z}$.
	}
\label{fig:DiscreteProcess}
\end{figure}

The stochastic process $\Process$ is generated by \emph{cylinder sets}:
\begin{align*}
U_{t,w} = \{ \MeasSymbol: \meassymbol_{t+1}, \ldots, \meassymbol_{t+\ell} = w\}
  ~,
\end{align*}
where $w \in \alphabet^\ell$ is a \emph{word} of length $\ell$. For a stationary
process, the \emph{word probabilities}:
\begin{align}
{\Pr}(\meassymbol_1 \ldots \meassymbol_\ell)
  = \mu (U_{ 0,\meassymbol_1 \ldots \meassymbol_\ell} )
\label{eq:FromMeasureToProb}
\end{align}
are sufficient to uniquely define the measure $\mu$.

That is, indexing is temporal and denoted by the use of subscripts $\TIndices$.
It is convenient to write a realization of $\{\MeasSymbol_t \}_{t \in
\mathbb{Z}}$ as an indexed sequence. For example, we write $\MeasSymbol_t =
\meassymbol$ to say that $\meassymbol \in \alphabet$ is the specific value of
$\MeasSymbol$ at time $t$. On occasion, we shorten this to $\meassymbol_t$.
That is, uppercase (e.g., $\MeasSymbol_t$) indicates the variable and lowercase
(e.g., $\meassymbol_t$) a realized value.

Blocks of consecutive RVs, called \emph{words}, are denoted by
$\MeasSymbol_{n:m} = \left\{ \MeasSymbol_t : n < t \leq m; n, m \in \integers
\right\}$ with the left index inclusive and the right exclusive. For example,
$\MeasSymbol_{0:3} = \MeasSymbol_0 \MeasSymbol_1 \MeasSymbol_2$. A word may also
refer to a particular realization of a given length, denoted in lowercase.  For
instance, one might write $\meassymbol_{0:3} = \meassymbol_0 \meassymbol_1
\meassymbol_2$ or $\meassymbol_{0:3} = bac$, if (say) $\meassymbol_t
\in \alphabet = \{a,b,c\}$.

\subsection{Process Dynamics}

Stochastic processes themselves can evolve over time. When working with these
dynamics, we need to refer to a process relative to a particular time $t$. To
do this we use superscripts. Denote the process referenced to time $t$ as
$\MeasSymbol^t = \{\MeasSymbol_t \}_{t \in \mathbb{Z}}$ and that referenced to
time $t+1$ as $\MeasSymbol^{t+1} = \{\MeasSymbol_{t+1} \}_{t \in \mathbb{Z}}$.

A natural dynamic for a stochastic process is given by the \emph{shift operator}:
$\shiftOperator : \alphabet^\TIndices \to \alphabet^\TIndices$ that simply
advances time $t \to t+1$:
\begin{align*}
(\shiftOperator \MeasSymbol)_t = \MeasSymbol_{t+1}
\label{eq:ShiftOperator}
  ~.
\end{align*}
This describes the shift acting on a process' individual RVs.
For the shift acting on the entire process we have:
\begin{align*}
\{ (\shiftOperator \MeasSymbol)_t \}_{t \in \mathbb{X}}
	= \{ \MeasSymbol_{t+1} \}_{t \in \mathbb{Z}} 
\end{align*}
or, more simply:
\begin{align*}
\shiftOperator \MeasSymbol^t = \MeasSymbol^{t+1}
	~.
\end{align*}

The shift also acts on measures over $\alphabet^\TIndices$:
\begin{align*}
(\shiftOperator \mu)A = \mu(\shiftOperator^{-1} A)
  ~,
\end{align*}
for measurable $A \in \sigmaAlgebra$.

A stochastic process paired with the shift operator becomes a dynamical system
$(\alphabet^\TIndices,\Sigma,\mu,\shiftOperator)$. A stochastic process is
stationary if the measure is time-shift invariant: $\shiftOperator \mu = \mu$.
It is \emph{ergodic} if, for all shift-invariant sets $\shiftOperator
\mathcal{I} = \mathcal{I}, ~\mathcal{I} \subset \alphabet^\TIndices$, either
$\mu(\mathcal{I}) = 0$ or $1$. Said more simply, an ergodic stochastic process
cannot be decomposed into other ergodic components. (Later on, we return to
notions of stationarity and ergodicity appropriate for information processes
and cognitive agents.)

When considering random variables and their probabilities, we continue, as
above, to denote random variables by capital Latin letters and specific
realizations by lower case. For example, $\Pr(X_t) = \{ \Pr(X_t = x) : x \in \{
1, \ldots , k\} \}$ and $\Pr \left( \MeasSymbol = \meassymbol \right) = \measure
\left( \left\{ \meassymbol \right\} \in \alphabet \right)$. 

\subsection{Sliding Window Processes}
\label{sec:SlidingProcesses}

It will be helpful to select a subset of related (time local) RVs in a process
via a function $h(\cdot)$ that gives offsets for the indices of those RVs:
\begin{align*}
h_r(t) = \{ t-r, \ldots, t-2, t-1, 0, t+1, t+2, \ldots t+r \}
  ~.
\end{align*}
In this way, one extracts a new \emph{sliding window} stochastic process
$(Y_t)_{t \in \mathbb{Z}}$ composed of a series of groups of the original
process' RVs. For example, if $r = 1$:
\begin{align*}
Y_t & = X_{[h_1(t)]} \\
    & = X_{t-1} X_{t} X_{t+1} ~\text{and} \\
Y_{t+1} & = X_{[h_1(t+1)]} \\
    & = X_{t} X_{t+1} X_{t+2}
  ~.
\end{align*}

\subsection{An Experiment and Its Realizations}
\label{sec:ExpReal}

An agent interacts with its environment by performing an \emph{experiment}
$\mathcal{E} = \{\omega^i \in \Sigma: i = 1, 2, \ldots, M\}$ consisting of a
number $M$ of realizations $\omega^i$. The agent prepares the environment
appropriately and initiates an experimental \emph{run} $\omega^i$ to generate
an arbitrarily long realization $\ms{0}{\infty} = \meassymbol_0, \ldots,
\meassymbol_\infty$, where each is a sequence of individual observations
$\meassymbol_t \in \alphabet$. In this way, an experiment is the set
$\mathcal{E} = \{ \omega^1 = \meassymbol^\prime_{0:\infty}, \omega^2 =
\meassymbol^{\prime\prime}_{0:\infty}, \ldots, \omega^M =
\meassymbol^{\prime\prime\prime}_{0:\infty} \}$ consisting of all of the
data---process realizations---available to an agent.

\subsection{Stationary}
\label{sec:Stationary}

Intuitively, the statistics of a stationary process do not depend on time.

\begin{definition}
\label{Def:Stationary}
In a \emph{stationary} process the measure is time-shift
invariant---$\shiftOperator \mu = \mu$---such that:
\begin{align*}
\Pr(\MS{t}{t+\ell}) = \Pr(\MS{0}{\ell})
  ~,
\end{align*}
for all $t \in \mathbb{\MeasSymbolVarTwo}$, $\ell \in
\mathbb{\MeasSymbolVarTwo}^+$.
\end{definition}

For a stationary process, the index denoting the present can nominally be set to
any value without altering subsequent analysis. 

\subsection{Ergodic}
\label{sec:Ergodic}

Intuitively, any realization $\omega^i \in \Sigma$ of an ergodic process has the same
statistical properties as any other realization $\omega^{j\neq i}$. More
formally, define the \emph{ensemble average} $E[f]$ of a measurable function
$f: \MeasSymbol^\mathbb{Z} \to \mathbb{R}$:
\begin{align*}
E[f] = \int_{\omega \in \MeasSymbol^\mathbb{Z}} f(\omega) d\mu(\omega) 
  ~.
\end{align*}

A function of a sample is often referred to as a \emph{statistic}. And so,
$E[f]$ is a statistic of the process samples $\omega^i$ for some function
$f(\cdot)$.

Compare this to the time average of the function:
\begin{align*}
\langle f(\omega) \rangle_t = \lim_{t \to \infty} \frac{1}{t} \sum_{k = 1}^t f(\tau^{k-1} \omega)
  ~,
\end{align*}
for $\mu$-almost every $\omega \in \alphabet^\mathbb{Z}$.

\begin{definition}
\label{Def:Ergodic}
An \emph{ergodic process} consists of realizations for which time-average
statistics equal ensemble averages:
\begin{align*}
\langle f(\omega) \rangle_t = E[f(\omega)]
  ~,
\end{align*}
where $f(\cdot)$ is any measurable function
and for $\mu$-almost every $\omega \in \alphabet^\mathbb{Z}$.
\end{definition}

\begin{figure}
\includegraphics[width=\columnwidth]{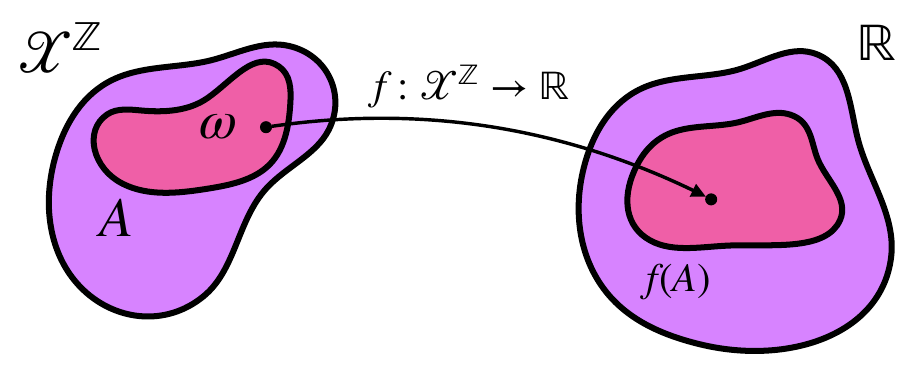}
\caption{Real-valued measurable function of a stochastic process: $f:
\alphabet^\mathbb{Z} \to \mathbb{R}$.
	}
\label{fig:FOfProcess}
\end{figure}

\begin{figure*}
\includegraphics[width=.99\textwidth]{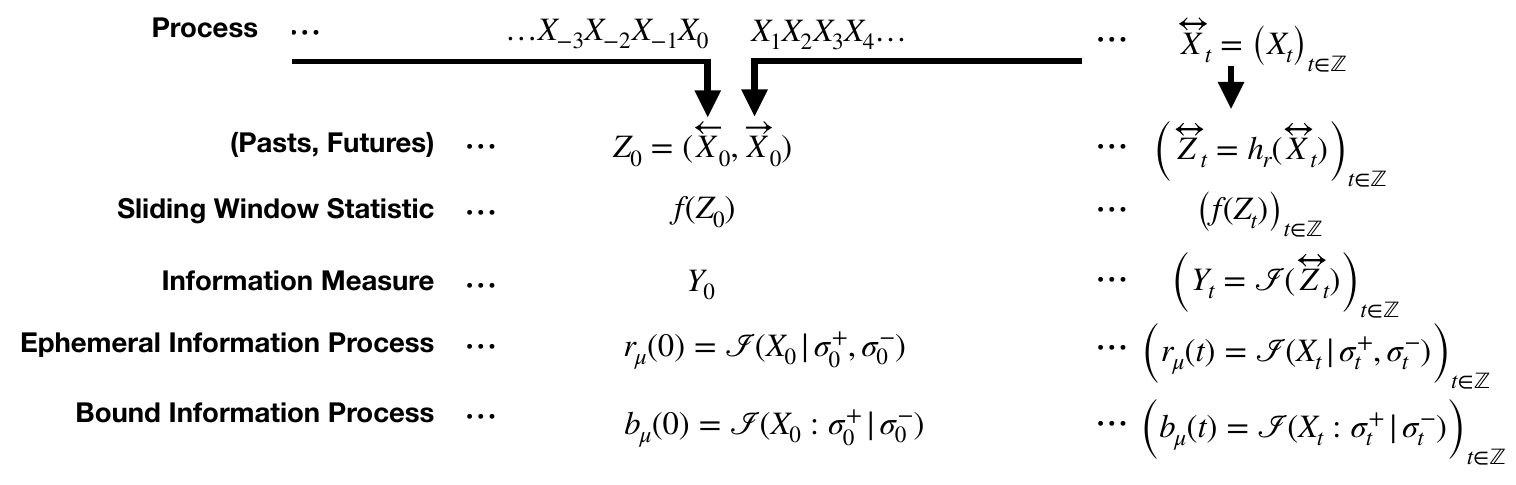}
\caption{Sliding window processes: The given stochastic process appears on the
	top line with the construction of sliding window processes from it
	illustrated going line by line. The last lines illustrate the ephemeral
	information process $\rmu (t)$ and the bound information process $\bmu(t)$.
	Recall that the forward and reverse causal states are functions of the
	semi-infinite past and futures, respectively---$\sigma^+ = \epsilon(\Past)$
	and $\sigma^- = \epsilon(\Future)$---which aggregate RVs are available in
	the $Z_t$ process.
	}
\label{fig:SlidingWindow}
\end{figure*}

\subsection{Functions of a Stochastic Process}
\label{sec:Process}

Recall that if $Z$ is a random variable and $f(\cdot)$ a real-valued,
measurable function, $f(Z)$ is a random variable with a real-valued event
space. Thus, since probability is such a function of a random variable, $Y =
\Pr(Z)$ is also a real-valued random variable with event space $\mathcal{Y} =
[0,1]$ \cite{Fell70a}. The following relies on this observation repeatedly.

Helpfully, this observation extends to sequences $\ldots \MeasSymbol_{-1}
\MeasSymbol_{0} \MeasSymbol_{1} \ldots$ of random variables that define a
stochastic process $\MeasSymbol = \{\MeasSymbol_t: t \in \mathbb{Z}\}$.
The main lesson is that a function of a stochastic process is a stochastic
process. See App. \ref{app:FofP} which considers real-valued functions $f(w)$
of realizations: $f: \MeasSymbol^\mathbb{Z} \to \mathbb{R}$.

The (push-forward) measure of a function of set $A \subset
\alphabet^\mathbb{Z}$ is:
\begin{align*}
\mu(f(A)) = \mu (f^{-1}(A))
~.
\end{align*}
(See Fig. \ref{fig:FOfProcess}.)

%
%

We are interested in an agent that observes a process through a finite-duration
window and manipulates that information at each time. Define a \emph{finite
range} $r$, real-valued function $f: \mathbb{R}^{2r+1} \to \mathbb{R}$, $r <
\infty$.

\begin{proposition}
Applying $f$ to a stochastic process $X$ gives a new stochastic process $Y =
\{Y_t: t \in \mathbb{Z}\}$, the components of which are:
\begin{align*}
Y_t = f \left( \MeasSymbol_{[h_r(t)]} \right)
  ~.
\end{align*}
\end{proposition}

\begin{proof}
$Y = f(X)$'s probability space $(\mathcal{Y}, \Sigma_Y, \mu_Y) =
(\mathbb{R}^\mathbb{Z}, \Sigma_Y, \mu_{Y})$.  $Y$'s event space $\Sigma_Y$ is
the sigma algebra over $\mathbb{R}^\mathbb{Z}$.  $Y$'s measure $\mu_Y$ is given
in terms of $\mu_X$ in two steps: giving the process $Y_t$ via a function of
RVs in a window onto $X$ and then determining the measure $\mu_Y$ in terms of
$\mu_X$.

First, consider the sliding window process $Z = h_r(X)$, where $h_r$ is a
finite-range $r$, index function as above. That is, $Z_t$'s event space consists
of the real-valued vectors $z = (x_0,x_1,\ldots, x_{2r+1}) \in
\mathbb{R}^{2r+1}$. Then the process $Z$ is a time series of vectors $Z_t =
(X_{t-r}, \ldots, X_t, \ldots X_{t+r})$.

And, second, we get $Z$'s measure $\mu_{t,z}$ from $X$'s $\mu_X$:
\begin{align*}
\mu_{t,z} = \int_{\omega \in U_{t,z}} d \mu_X (\omega)
  ~,
\end{align*}
with the cylinder $U_{t,z} = \{\MeasSymbol: (\meassymbol_t, \meassymbol_{t+1},
\ldots, \meassymbol_{2r+1}) = z\}$ associated with vector $z = (z_0, z_1,
\ldots, z_{2r+1})$. Since $\MeasSymbol$ is stationary we can set $t = 0$ and
work with $U_0,z$; simplifying notation to $\mu_Z$ and $U_z$. 

Next, we obtain the real-valued process $Y_t$ from the vector-valued process
$Z_t$. However, this is simply a componentwise transformation of the vector
series by the real-valued function $f(\cdot)$: $Y_t = f(Z_t)$.

Finally, we determine the measure $\mu_Y$ for the real-valued process $Y_t$
from the measure $\mu_Z$ for the vector-valued process $Z_t$. Recall the
real-valued function $f : \mathbb{R}^{2r+1} \to \mathbb{R}$ that maps real
vectors $z$ to some real value: $f(z)$. To get $Y$'s measure $\mu_Y$ in
terms of $\mu_Z$ we integrate over the set of vectors $z \in f^{-1} (y)$,
that give the real value $y$:
\begin{align}
\mu_Y (y) = \int_{z \in f^{-1} (y)} d \mu_Z (z)
  ~,
\end{align}
where $y = \ldots y_{t-1} y_t y_{t+1} \ldots \in \mathbb{R}^\mathbb{Z}$ is a
realization of $Y$. And so, $Y = f(Z)$'s probability space is $(\mathbb{R}^\mathbb{Z}, \Sigma_Y, \mu_{Y})$.

Together, these determine $Y$'s probability space in terms of $\MeasSymbol$'s
probability space $\left(\mathcal{\MeasSymbol}, \Sigma_\MeasSymbol,
\mu_\MeasSymbol\right)$ establishing that $Y = f \circ h (\MeasSymbol) =
f(\MeasSymbol)$ is a stochastic process with measure:
\begin{align*}
\mu_Y (y) =
\int_{\omega \in U_{t,z}}
\int_{z \in f^{-1} (y)}
d \mu_Z (z)
d \mu_X (\omega)
  ~.
\end{align*}
\end{proof}

Moreover, if $\MeasSymbol$ is stationary and ergodic, so is process $Y$. As the
following two results establish.

\begin{proposition}
\label{Def:StationaryFunction}
If $\MeasSymbol$ is a stationary process and $f$ is a finite-range, real-valued
function, then $Y = f(\MeasSymbol)$ is stationary.
\end{proposition}

\begin{proof}
We must show that $Y$ is time-shift invariant: $\Pr(Y_{t:t+\ell}) =
\Pr(Y_{0:\ell})$, for all $t \in \mathbb{Z}$. We are given that $\MeasSymbol$ is
time-shift invariant: $\Pr(X_{t:t+\ell}) = \Pr(X_{0:\ell})$, for all $t,\ell \in
\mathbb{Z}$. And so, the process over its $2r+1$ blocks $\MeasSymbol_{0:2r+1}$
is time-shift invariant. Since $Y$ values are direct functions of these blocks,
$Y$'s sequences are time-shift invariant and so $Y$ is stationary.
\end{proof}

\begin{corollary}
If $\MeasSymbol$ is a stationary process, the \emph{probability process}
$Y = \Pr(\MeasSymbol)$ is stationary.
\end{corollary}

\begin{proof}
$X$'s sequence probabilities $Y = \Pr(\MeasSymbol)$ are RVs. $X$'s stationarity
guarantees that a time series of its sequence probabilities $\ldots Y_{t-1} Y_t
Y_{t+1} \ldots$ is time-shift invariant. However, sequence probabilities $Y_t$
are obtained from a finite-range, real-valued function $\Pr(\cdot)$. And so, by
the preceding proposition, the stochastic process $Y = \Pr(\MeasSymbol)$ is
stationary.
\end{proof}

Moreover:

\begin{proposition}
\label{Def:ErgodicFunction}
If $\MeasSymbol$ is a stationary and ergodic process and $f$ is a finite-range,
real-valued function, then $Y = f(\MeasSymbol)$ is stationary and ergodic.
\end{proposition}

\begin{proof}
(See
Ref. \cite[Thm. 36.4]{Bill95a}.) We
have process $Y = \{(Y)_t: t \in \mathbb{Z}\}$, with $Y_t = f \left(
\MeasSymbol_{t-r}, \ldots, \MeasSymbol_t, \ldots, \MeasSymbol_{t+r} \right)$.
$Y$ is stationary from the preceding proposition. We must show that $Y$ is
ergodic:
\begin{align*}
\langle g(Y) \rangle_t = E[g(\omega)]
  ~,
\end{align*}
for $\omega \in \Sigma_Y$ and functions $g$. We are given that this holds for
the original process $\MeasSymbol$:
\begin{align*}
\langle h(\MeasSymbol) \rangle_t = E[h(x)]
  ~,
\end{align*}
for $x \in \MeasSymbol^\mathbb{Z}$ and all functions $h$. Thus, it also holds
for $h = f$:
\begin{align*}
\langle f(\MeasSymbol) \rangle_t & = E[f(x)] \\
\langle Y \rangle_t & = E[\omega]
  ~,
\end{align*}
and the expectation is taken over all $\omega \in \MeasSymbol^\mathbb{Z}$ and $\Sigma_Y = f(\Sigma_X)$.
\end{proof}

\begin{corollary}
If $\MeasSymbol$ is stationary and ergodic, then the probability process
$\Pr(\MeasSymbol)$ is stationary and ergodic.
\end{corollary}

\begin{proof}
We are given that $X$'s sequence probabilities are time-shift invariant and that
it has only a single ergodic invariant set. $Y$'s sequence probabilities are
obtained from $\MeasSymbol$ via a finite-range, real-valued function
$\Pr(\cdot)$. And so, the preceding proposition shows that the process
$\Pr(\MeasSymbol)$ is stationary and ergodic.
\end{proof}

Figure \ref{fig:SlidingWindow} illustrates the construction of the
sliding window processes from which information processes are constructed.

\begin{table}[]
	\setlength{\tabcolsep}{0pt}
    \begin{tabular}{rccccc}
    \toprule
    \multicolumn{2}{l}{$\Leftarrow$ Reverse} &
    \multicolumn{2}{c}{Processes} &
    \multicolumn{2}{r}{Forward $\Rightarrow$} \\
    \midrule
	Reverse Statistical &
	$~~\cdots~~$ &
	$\SelfII{\causalstate_{-1}^-}$ &
	$\SelfII{\causalstate_{0}^-}$ &
	$\SelfII{\causalstate_{1}^-}$ &
	$~\cdots~$\\
	Complexity & $~\cdots~$ &
	$\Cmu^-(-1)$ &
	$\Cmu^-(0)$ &
	$\Cmu^-(1)$ &
	$~\cdots~$\\
    \midrule
	Reverse&
	$\cdots$ &
	$~~\epsilon^- (\Future_{-1})$ &
	$~~\epsilon^- (\Future_{0})$ &
	$~~\epsilon^- (\Future_{1})$ &
	$~\cdots~$\\
	Causal Process & $~~\cdots~~$ &
	$\causalstate^-_{-1}$ &
	$\causalstate^-_{0}$ &
	$\causalstate^-_{1}$ &
	$~\cdots~$\\
    \bottomrule
	Futures &
	$~\cdots~$ &
	$\Future_{-1}$ &
	$\Future_{0}$ &
	$\Future_{1}$ &
	$~\cdots~$\\
	Presents&
	$~\cdots~$ &
	$\MeasSymbol_{-1}$ &
	$\MeasSymbol_{0}$ &
	$\MeasSymbol_{1}$ &
	$~\cdots~$\\
	Pasts &
	$~\cdots~$ &
	$\Past_{-1}$ &
	$\Past_{0}$ &
	$\Past_{1}$ &
	$~\cdots~$\\
    \bottomrule
	Forward&
	$~\cdots~$ &
	$\epsilon^+ (\Past_{-1})$ &
	$\epsilon^+ (\Past_{0})$ &
	$\epsilon^+ (\Past_{1})$ &
	$~\cdots~$\\
	Causal Process & $~\cdots~$ &
	$\causalstate^+_{-1}$ &
	$\causalstate^+_{0}$ &
	$\causalstate^+_{1}$ &
	$~\cdots~$\\
    \bottomrule
	Forward Statistical &
	$~\cdots~$ &
	$\SelfII{\causalstate_{-1}^+}$ &
	$\SelfII{\causalstate_{0}^+}$ &
	$\SelfII{\causalstate_{1}^+}$ &
	$~\cdots~$\\
	Complexity & $~~\cdots~~$ &
	$\Cmu^+(-1)$ &
	$\Cmu^+(0)$ &
	$\Cmu^+(1)$ &
	$~\cdots~$\\
    \bottomrule
\end{tabular}
\caption{Information processes: The given stochastic process appears on the
	line with Presents. All other lines are information processes that
	derive from it.
	}
\label{fig:InformationProcesses}
\end{table}

\section{Informations and Information Processes}
\label{sec:InfoTh}

We introduce Shannon's information measures, information diagrams, and
information processes, showing how to use them to describe a variety
of informational properties of a stochastic process $\Process$.
Table \ref{fig:InformationProcesses} provides a roadmap for the various kinds
of stochastic informational processes the following introduces.
The companion, Table \ref{fig:SymmetricInformationProcesses}, lists several
information processes that are time symmetric in the sense that the sets of
RVs from which they are formed is symmetric under $t \to -t$.

\subsection{Self-Informations and Measures}
\label{sec:InformationMeasures}

The most basic quantity in information theory is the
\emph{self-information}---the amount of information learned observing a single
measurement value $\meassymbol$ of a random variable $\MeasSymbol_t$
\cite{Shan48a,Cove06a}:
\begin{align}
\selfii{\MeasSymbol_t = \meassymbol} =
    - \log_2 \Pr(\MeasSymbol_t = \meassymbol)
    ~.
\label{eq:SelfInformation}
\end{align}
On the one hand, if the occurrence of $\meassymbol$ is certain, an observer
learns no information and $\selfi = 0$. On the other hand, if all realizations
are equally uncertain $\Pr(\MeasSymbol_t = \meassymbol) = 1/k$ for each $x$,
the observer gains the maximum amount of information---that is, any one
occurrence is maximally informative and $\selfii{\MeasSymbol_t = \meassymbol} =
\log_2 k$.

Several comments on notation and terminology are useful initially. The prefix
``self-'' connotes a quantity describing an individual event's occurrence. This
deviates from the typically-rare use of ``self-'' in ``self-information'' as the
mutual information of a random variable with itself, which is simply its
entropy: $I[X:X] = H[X]$ in the notation of Ref. \cite[p. 21]{Cove06a}. In
addition, here within the theory of information measures \cite{Yeun08a}, there
is only a single quantity \emph{information}, denoted $\SelfI$. In this, one
has $\SelfImut{X}{X} = \SelfII{X}$, highlighting the redundancy in the
conventional use of distinct symbols $H$ and $I$. Different kinds of
information---joint, conditional, mutual---are then denoted by the form of
$\SelfI$'s arguments.  This leads to the practical consequence that we use only
the single notation $\SelfI[\cdot]$ rather than, as in basic information theory
\cite{Cove06a}, using both $H[\cdot]$ and $I[\cdot]$.

With reference to self-information, the more familiar quantity from information
theory, though, is its ensemble average---the Shannon entropy
$\SelfII{\MeasSymbol} = E[\selfii{\MeasSymbol_t}]$ of the random variable
$\MeasSymbol_t$:
\begin{align}
    \SelfII{\MeasSymbol_t} & =  E[\selfii{\MeasSymbol_t}] \nonumber \\
    & = - \sum_{\meassymbol \in \alphabet} 
    \Pr(\MeasSymbol_t = \meassymbol) \log_2 \Pr(\MeasSymbol_t = \meassymbol)
    ~.
\label{eq:ShannonEntropy}
\end{align}
This requires being given or knowing ahead of time the single-symbol average
probabilities $\{ \Pr(\MeasSymbol_t = \meassymbol), \meassymbol \in
\alphabet\}$. Given a stochastic process, it can also be developed from
measuring the self-information over time:
\begin{align*}
\langle \selfii{\MeasSymbol_t = \meassymbol} \rangle
  = \lim_{T \to \infty} \frac{1}{T} \sum_{t=0}^T
  \selfii{\MeasSymbol_t = \meassymbol}
  ~,
\end{align*}
assuming that the process is stationary and ergodic.

And, from this, the ensemble average is:
\begin{align*}
    \SelfII{\MeasSymbol_t} =
    - \sum_{\meassymbol \in \alphabet} 
	\langle \selfii{\MeasSymbol = \meassymbol} \rangle
    ~.
\end{align*}

Similarly, there is the relationship between a pair of jointly-distributed
random variables, say, $\MeasSymbol_t$ and $\MeasSymbol_{t^\prime}$. To
simplify, considered time-delayed RVs where $t^\prime = t + \tau$. The
\emph{joint entropy} $\SelfIjoint{\MeasSymbol_t}{\MeasSymbol_{t^\prime}}(\tau)$
is of the same functional form as \cref{eq:ShannonEntropy}, applied to the
joint distribution $\Pr \left( \MeasSymbol_t, \MeasSymbol_{t^\prime} \right)$.

First, we have the joint self-information:
\begin{align}
    \selfijoint{\MeasSymbol_t = \meassymbol}{\MeasSymbol_{t^\prime} =
	\meassymbol^\prime} (\tau) =
    - \log_2 \Pr(\MeasSymbol_t = \meassymbol,\MeasSymbol_{t^\prime} =
	  \meassymbol^\prime)
    ~,
\label{eq:JointSelfInformation}
\end{align}
for $\meassymbol, \meassymbol^\prime \in \alphabet$. Second, the
ensemble-averaged version:
\begin{align}
    & \SelfIjoint{\MeasSymbol}{\MeasSymbol^\prime} (\tau) = \nonumber \\
    & - \sum_{\meassymbol, \meassymbol^\prime \in \alphabet} 
	\Pr(\MeasSymbol_t = \meassymbol,\MeasSymbol_{t^\prime} =
      \meassymbol^\prime)
	  \log_2 \Pr(\MeasSymbol_t = \meassymbol,\MeasSymbol_{t^\prime} =
      \meassymbol^\prime)
    ~.
\label{eq:ShannonJointEntropy}
\end{align}

\begin{table}[]
	\setlength{\tabcolsep}{0pt}
    \begin{tabular}{lccccc}
    \toprule
    \multicolumn{6}{l}{Processes} \\
    \midrule
	Prediction & $~\cdots~$ &
	$\SelfIcond{\MeasSymbol_{-1}} {\causalstate_{-1}^+}$ &
	~$\SelfIcond{\MeasSymbol_{0}} {\causalstate_{0}^+}$ ~&
	$\SelfIcond{\MeasSymbol_{1}} {\causalstate_{1}^+}$ &
	$~\cdots~$\\
	& $~\cdots~$ &
	$\hmu^+(-1)$ &
	$\hmu^+(0)$ &
	$\hmu^+(1)$ &
	$~\cdots~$\\
    \midrule
	Predictable & $~\cdots~$ &
	$\SelfImut{\causalstate_{-1}^+} {\causalstate_{-1}^-}$ &
	~$\SelfImut{\causalstate_{0}^+} {\causalstate_{0}^-}$ ~&
	$\SelfImut{\causalstate_{1}^+} {\causalstate_{1}^-}$ &
	$~\cdots~$\\
	& $~\cdots~$ &
	$\EE(-1)$ &
	$\EE(0)$ &
	$\EE(1)$ &
	$~\cdots~$\\
    \midrule
	Ephemeral & $~\cdots~$ &
	$\rmu(-1)$ &
	$\rmu(0)$ &
	$\rmu(1)$ &
	$~\cdots~$\\
    \midrule
	Bound & $~\cdots~$ &
	$\bmu(-1)$ &
	$\bmu(0)$ &
	$\bmu(1)$ &
	$~\cdots~$\\
    \bottomrule
\end{tabular}
\caption{Time Symmetric Information Processes: Derived from
	information measures that are symmetric in time.
	}
\label{fig:SymmetricInformationProcesses}
\end{table}

For which we can also implement a time-averaged quantity:
\begin{align*}
\langle \selfijoint{\MeasSymbol_t = \meassymbol}{\MeasSymbol_{t^\prime} =
\meassymbol^\prime} \rangle (\tau)
  = \lim_{T \to \infty} \frac{1}{T} \sum_{t=0}^T
  \selfijoint{\MeasSymbol_t = \meassymbol}{\MeasSymbol_{t^\prime} = \meassymbol^\prime}
  ~,
\end{align*}
From this, the ensemble average is then given:
\begin{align*}
	\SelfIjoint{\MeasSymbol}{\MeasSymbol^\prime} (\tau) =
    - \sum_{\meassymbol,\meassymbol^\prime \in \alphabet} 
	\langle \selfijoint{\MeasSymbol_t = \meassymbol}{\MeasSymbol_{t^\prime} =
	\meassymbol^\prime} \rangle (\tau)
    ~.
\end{align*}
Effectively, via sampling, the time average ``empirically'' provides an estimate of the joint $\Pr(\MeasSymbol_t = \meassymbol,\MeasSymbol_{t^\prime}
= \meassymbol^\prime)$.

These can be straightforwardly extended, in principle, from pairs of variables
to the \emph{multivariate joint entropy} $\SelfI(\RVSet)$ of a set of $N$
variables $\RVSet = \left\{ \MeasSymbol_i \mid i \in \left( 1, \dots, N \right)
\right\}$. (In this, the temporal relationship between the variables needs to
be specified by an appropriate set of indexes; such as, time delays.)

Finally, the \emph{conditional entropy}
$\SelfIcond{\MeasSymbol_t}{\MeasSymbol_t^\prime} (\tau)$ is of a similar
functional form as above, but applied to the conditional distribution $\Pr
\left( \MeasSymbol_t| \MeasSymbol_{t^\prime} \right)$. It gives the information
learned from observing a random variable $\MeasSymbol_t$ at one time given
knowledge of random variable $\MeasSymbol_{t^\prime}$ at another $t^\prime = t
+ \tau$. We can write this in terms of the joint entropy:
\begin{align}
    \SelfIcond{\MeasSymbol_t}{\MeasSymbol_{t^\prime}} (\tau) = 
    \SelfIjoint{\MeasSymbol_t}{\MeasSymbol_{t^\prime}} (\tau)
	- \SelfII{\MeasSymbol_{t^\prime}}
	~. 
\label{eq:ConditionaLEntropy}
\end{align}

First, if we break this down as above, we have the conditional self-information:
\begin{align}
    \selficond{\MeasSymbol_t = \meassymbol}{\MeasSymbol_{t^\prime} =
	\meassymbol^\prime}
	=
    - \log_2 \Pr(\MeasSymbol_t = \meassymbol|\MeasSymbol_{t^\prime} =
	  \meassymbol^\prime)
    ~,
\label{eq:ConditionalSelfInformation}
\end{align}
for $\meassymbol, \meassymbol^\prime \in \alphabet$. Second, there is the
ensemble-averaged version:
\begin{align}
    & \SelfIcond{\MeasSymbol}{\MeasSymbol^\prime} (\tau) = \nonumber \\
    & - \sum_{\meassymbol, \meassymbol^\prime \in \alphabet} 
	\Pr(\MeasSymbol_t = \meassymbol,\MeasSymbol_{t^\prime} =
      \meassymbol^\prime)
	  \log_2 \Pr(\MeasSymbol_t = \meassymbol|\MeasSymbol_{t^\prime} =
      \meassymbol^\prime)
    ~.
\label{eq:ShannonConditionalEntropy}
\end{align}
For which we can also implement a time-averaged quantity:
\begin{align*}
\langle & \selficond{\MeasSymbol_t = \meassymbol}{\MeasSymbol_{t^\prime} =
\meassymbol^\prime} \rangle (\tau) \\
  & \qquad= \lim_{T \to \infty} \frac{1}{T} \sum_{t=0}^T
  \selficond{\MeasSymbol_t = \meassymbol}{\MeasSymbol_{t^\prime} =
  \meassymbol^\prime)}
  ~,
\end{align*}
If the stochastic process is stationary and ergodic,
then the ensemble average is given:
\begin{align*}
	\SelfIcond{\MeasSymbol}{\MeasSymbol^\prime} (\tau) =
    - \sum_{\meassymbol,\meassymbol^\prime \in \alphabet} 
	\langle \selficond{\MeasSymbol_t = \meassymbol}{\MeasSymbol_{t^\prime} =
  \meassymbol^\prime)} \rangle (\tau)
    ~.
\end{align*}


The information shared between random variables is the \emph{mutual
information}. The mutual self-information of observing $\meassymbol$ and
$\meassymbol^\prime$ at times $t$ and (delayed) $t^\prime$, respectively, is:
\begin{align}
	\selfimut{\MeasSymbol_t = \meassymbol}{ \MeasSymbol_{t^\prime} =
	\meassymbol^\prime}
	= 
    \log_2 \left( \frac{\Pr \left( \MeasSymbol_t = \meassymbol, 
    \MeasSymbol_{t^\prime} = \meassymbol^\prime \right) }{\Pr 
    \left( \MeasSymbol_t = \meassymbol \right) 
    \Pr \left( \MeasSymbol_{t^\prime} = \meassymbol^\prime \right)} \right)
	~. 
\label{eq:OnlineMutualSelfInformation}
\end{align}
When appropriately ensemble-averaged this recovers the version familiar from
elementary information theory:
\begin{align}
    & \SelfImut{\MeasSymbol_t}{\MeasSymbol_{t^\prime}} (\tau) \nonumber \\
	& = \sum_{ \substack{
	\meassymbol^\prime \in \alphabet \\
    \meassymbol \in \alphabet} }
	\Pr \left( \MeasSymbol_t = \meassymbol, 
    \MeasSymbol_{t^\prime} = \meassymbol^\prime \right)
	\selfimut{\MeasSymbol_t = \meassymbol}{ \MeasSymbol_{t^\prime} =
    \meassymbol^\prime}
	.
\label{eq:MutualSelfInformation}
\end{align}
And, a time-averaged version is:
\begin{align}
    \langle
	& \selfimut{\MeasSymbol_t = \meassymbol}{ \MeasSymbol_{t^\prime + \tau} =
    \meassymbol^\prime}
	\rangle (\tau)
	\nonumber \\
	& \quad = \lim_{T \to \infty} \frac{1}{T}
	\sum_{t = 0}^T
	\selfimut{\MeasSymbol_t = \meassymbol}{ \MeasSymbol_{t^\prime} =
    \meassymbol^\prime}
	~.
\label{eq:TemporalMutualSelfInformation}
\end{align}
Again, if the stochastic process is stationary and ergodic, then the
time-average and ensemble average are equal.
And, the ensemble average is given:
\begin{align*}
	\SelfImut{\MeasSymbol_t}{\MeasSymbol_{t^\prime}} (\tau) =
    - \sum_{\meassymbol,\meassymbol^\prime \in \alphabet} 
	\langle \selfimut{\MeasSymbol_t = \meassymbol}{\MeasSymbol_{t^\prime + \tau} =
  \meassymbol^\prime} \rangle (\tau)
    ~.
\end{align*}

The mutual information $\operatorname{I} \left[ \MeasSymbol_{m,n} :
\MeasSymbol_{p,q} : \MeasSymbol_{r,s} \right]$ between three random-variable
blocks---known as the \emph{interaction information} or as one of the
\emph{multivariate mutual informations}---is given by the difference between
mutual information and conditional mutual information:
\begin{align}
\operatorname{I} & \left[ \MeasSymbol_{m,n} : \MeasSymbol_{p,q} :
	\MeasSymbol_{r,s} \right] \nonumber \\
	& = 
    \I{\MeasSymbol_{m,n}}{\MeasSymbol_{r,s}} - 
    \Icond{\MeasSymbol_{m,n}}{\MeasSymbol_{p,q}}{\MeasSymbol_{r,s}}
	~. 
\label{eq:interactionInfo}
\end{align}

Tracking three-way interactions in this way between variables $X$, $Y$, and $Z$,
say, brings up two interpretations worth highlighting here. Two variables
$\MeasSymbol$ and $Y$ can have positive mutual information but be conditionally
independent in the presence of $Z$, in which case the interaction information is
positive. It is also possible, though, for two independent variables to become
correlated in the presence of $Z$, making the conditional mutual information
positive and the interaction information negative.

In other words, conditioning
on a third variable $Z$ can either increase or decrease mutual information and
the $X$ and $Y$ variables can appear more or less dependent given additional
observations \cite{Cove06a}. That is, there can be \emph{conditional
independence} or \emph{conditional dependence} between a pair of random
variables. Note that the interaction information is symmetric, so these
interpretations hold regardless of the conditioning variable selected. 

Finally, note that the self-information functions $\selfi(\cdot)$ above:
\begin{enumerate}
      \setlength{\topsep}{-4pt}
      \setlength{\itemsep}{-4pt}
      \setlength{\parsep}{-4pt}
\item $-\log_2 \Pr(\MeasSymbol_t)$,
\item $-\log_2 \Pr(\MeasSymbol_t,\MeasSymbol_{t^\prime})$,
\item $-\log_2 \Pr(\MeasSymbol_t | \MeasSymbol_{t^\prime})$, and 
\item $-\log_2 \left( \Pr(\MeasSymbol_t \MeasSymbol_{t^\prime}) /
\Pr(\MeasSymbol_t) \Pr(\MeasSymbol_{t^\prime}) \right)$ .
\end{enumerate}
are used in weighted averages to give ensemble averages. Their temporal and
ensemble averages are equal when the process in question is stationary and
ergodic.

In this, these self-informations lead to distinct and complementary
informational quantities useful in characterizing a process' various properties.
(This will become apparent shortly.) And, in the role they play in that
averaging, they are occasionally (suggestively) referred to as ``densities''
\cite{Pins64a}. This language parallels that found in continuous random
variable theory whose ``densities'' are integrated to obtain probability
``distributions''.

Note that the second pair of self-informations are linear functions of the first
two:
\begin{align*}
& -\log_2 \Pr(\MeasSymbol_t | \MeasSymbol_{t^\prime})
  = -\log_2 \Pr(\MeasSymbol_t,\MeasSymbol_{t^\prime})
  - \log_2 \Pr(\MeasSymbol_t) \\
&  \text{and:} \\
  & -\log_2 \left( \Pr(\MeasSymbol_t \MeasSymbol_{t^\prime}) /
\Pr(\MeasSymbol_t) \Pr(\MeasSymbol_{t^\prime}) \right) \\
  & \qquad= -\log_2 \Pr(\MeasSymbol_t,\MeasSymbol_{t^\prime})
  + \log_2 \Pr(\MeasSymbol_t) + \log_2 \Pr(\MeasSymbol_{t^\prime})
  ~.
\end{align*}
This also holds true \cite[Sec. 3]{Yeun08a} for other Shannon measures that we
develop shortly. And so, in the following we need only consider the ergodic
properties of the marginal and joint temporal self-informations, due to the
linearity of the conditions for stationarity and ergodicity.

Taken altogether, the preceding simply reviewed quantities from elementary
information theory, while pointing out how they are developed from
time-dependent self-informations (the ``densities'') via ensemble or temporal
averaging. These densities play a key role later when describing how an agent,
interacting in real-time with a complex environment, generates a time series---a
stochastic process of---time-local self-informations.

\subsection{Temporal Information Atoms}
\label{sec:TemporalAtoms}

One of the most basic operations of a stochastic process is to create
information \cite{Kolm59,Sina59}---measured by the entropy density (or rate)
$\hmu$. Importantly, in most settings there is more going on than creation.
Reference \cite{Jame11a} showed that created information consists of two
parts: ephemeral information $\rmu$ is that part of the created information
which the present throws away and the bound information $\bmu$ is that part
which the process remembers---that can affect the future \cite{Jame13a}. Said
simply, processes \emph{both} forget and remember portions of the information
created at each time: $\hmu = \rmu + \bmu$. This realization led to a more
detailed decomposition and identification of the kinds of information and
transformations with which stochastic processes operate.

Beyond time series, in the more general circumstance of a set of RVs,
information can be generally decomposed into a collection of atoms; see Ref.
\cite{Jame17a}. In the time series setting, though, the RVs are related by a
time shift and this simplifies understanding the various kinds of information
available in a stochastic process. We refer to the observation $\MeasSymbol_t =
\meassymbol$ at time $t$ as occurring in the \emph{present}. We call the
semi-infinite sequence $\MeasSymbol_{-\infty : t }$ the \emph{past} at time
$t$, which we also (more frequently) denote with a left pointing arrow
$\PastSmashed_t$. Accordingly, the semi-infinite sequence $\MeasSymbol_{t+1 :
\infty}$ is called the \emph{future} at time $t$ and denoted
$\FutureSmashed_t$. Thus, our strategy for analyzing the information theory of
stochastic processes is primarily concerned with delineating the relationships
between the past, present, and future. In this, we follow Ref. \cite{Jurg25a}.

Given this strategy, a useful perspective on a stochastic process is to picture
it as a communication channel transmitting information from the past
$\Past_{t-1} = \dots \MSym_{t-3}, \MSym_{t-2}, \MSym_{t-1}$ to the future
$\Future_t = \MSym_{t+1}, \MSym_{t+2}, \MSym_{t+3} \dots$ through the medium of
the present $\MeasSymbol_t$.
%
%
The past and the future are depicted in \cref{fig:iDiagramProcess}'s
\emph{information diagram} as extending to the left and the right,
respectively, to emphasize visualizing the bi-infinite chain of random
variables \cite{Jurg25a}.

\begin{figure*}
\includegraphics[width=.99\textwidth]{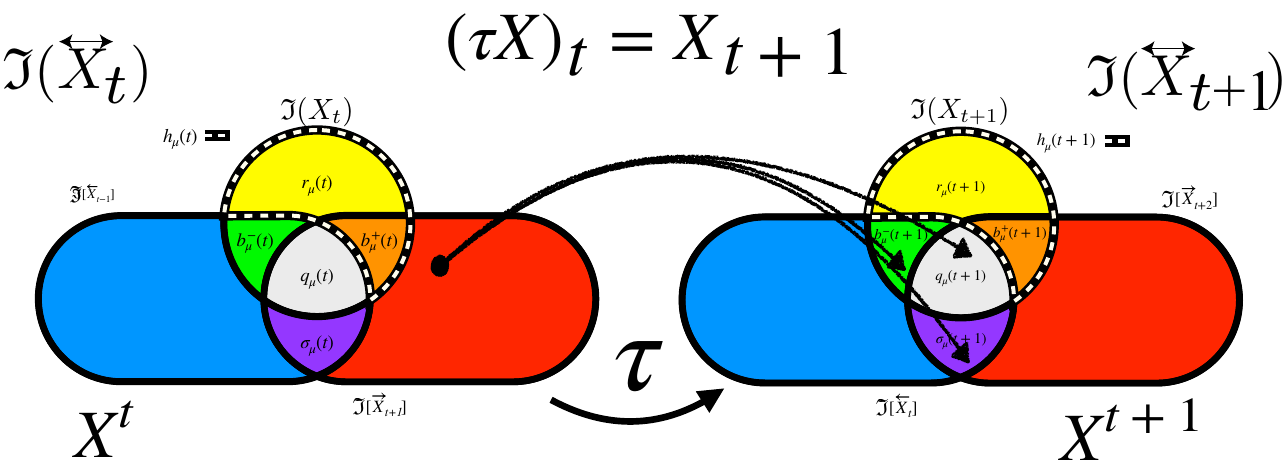}
\caption{Information processes $\SelfI[A|\overline{A}]$ when shift operator
	$\tau$ (Eq. \eqref{eq:ShiftOperator}) acts on stationary, ergodic
	stochastic process $\{ \MeasSymbol_t \}_{t \in \mathbb{Z}}$ for atoms $A$
	in Table \ref{tab:InfoAtomsProcess}. Temporal information diagrams
	representing how the informational relationships between the stochastic
	process indexed at $t$ evolve into those indexed to $t+1$:
	$\MeasSymbol^{t+1} = \shiftOperator(\MeasSymbol^{t})$---how
	the future $\Future_t$, the present $\MeasSymbol_t$, and the past
	$\Past_{t-1}$ evolve over time to the next future $\Future_{t+2}$, present
	$\MeasSymbol_{t+1}$, and past $\Past_{t}$. (Left) I-diagram at time $t$,
	formally denoted $\SelfI(\BiInfinity_t)$; (Right) I-diagram at time $t+1$,
	formally denoted $\SelfI(\BiInfinity_{t+1})$. The i-diagrams are labeled
	with the time $t$ and $t+1$ ephemeral informations $\rmu(t)$ and
	$\rmu(t+1)$, the binding informations $\bmu^\pm (t)$ and $\bmu^\pm (t+1)$,
	the enigmatic informations $\qmu(t)$ and $\qmu(t+1)$, the elusive
	informations $\sigmu(t)$ and $\sigmu(t+1)$. The Shannon entropy rates
	$\hmu(t) = \rmu(t) + \bmu(t)$ and $\hmu(t+1) = \rmu(t+1) +
	\bmu(t+1)$---composite information atoms---are outlined with a white-dashed
	line.
	}
\label{fig:iDiagramProcess}
\end{figure*}

One might expect increasing difficulty when moving from one or several isolated
random variables to a stochastic process of semi-infinite random-variable
chains. However and, at least initially, profiling a process' information atoms
in terms of its past $\PastSmashed_t$, present $\MeasSymbol_t$, and future
$\FutureSmashed_t$ requires little more information-theoretic set-up than
that already given; see also Ref. \cite{Jurg25a}.


Given these three random variables in play at each time $t$,
there is a set of eight quantities---\emph{information atoms}. These are shown
in information-diagram form in \cref{fig:iDiagramProcess}(left) and 
are collected in \cref{tab:InfoAtomsProcess}. The related aggregate measures,
including $\hmu$ already described, are also listed. Each is interpreted in
more detail shortly.

\begin{table}[]
	\setlength{\tabcolsep}{0pt}
    \begin{tabular}{ccc}
    \toprule
    \multicolumn{3}{c}{Information Atoms} \\
	$~~A$ & $\overline{A}$ & $\SelfI[A|\overline{A}]$ \\
    \midrule
    $\Big\{ \MeasSymbol_t \Big\}$ & $\Big\{ \Past_t, \Future_t \Big\}$
	& $\rmu(t)$ \\
    $\left\{ \Past_t: \MeasSymbol_t \right\} $ & $\Big\{ \Future_t
    \Big\}$ & $\bmureverse(t)$ \\
    $\left\{ \Past_t: \Future_t \right\}$ & $\Big\{ \MeasSymbol_t
    \Big\}$ & $\sigmu(t)$\\
    $\left\{ \MeasSymbol_t: \Future_t \right\}$ & $\Big\{ \Past_t
    \Big\}$ & $\bmuforward(t)$\\
    $\left\{ \Past_t: \MeasSymbol_t: \Future_t \right\}$ & &
    $\qmu(t)$\\
    \toprule
    \multicolumn{3}{c}{Composite Information Measures} \\
	$A$ & $\overline{A}$ & $\SelfI[A|\overline{A}]$ or $\SelfI[A:\overline{A}]$
	 ~(*)\\
    \midrule
    $\left\{ \MeasSymbol_t \right\}$ &
    $\left\{\Past_t\right\}$ & $\hmu^+(t) = \rmu(t) + \bmuforward(t)$ \\
    $\left\{ \MeasSymbol_t \right\}$ &
    $\left\{\Future_t\right\}$ & $\hmu^-(t)  = \rmu(t) + \bmureverse(t)$ \\
    $\left\{ \MeasSymbol_t \right\}$ &
    $\left\{\Past_t\right\}$ & $\rhomu^+(t) = \qmu(t) + \bmureverse(t)$ ~(*) \\
    $\left\{ \MeasSymbol_t \right\}$ &
    $\left\{\Future_t\right\}$ & $\rhomu^-(t)  = \qmu(t) + \bmuforward(t)$ ~(*) \\
    $\left\{ \Past_t:\Future_t \right\}$ & & $\EE(t) = \bmuforward(t) +
	\sigmu(t) + \qmu(t) ~(*) $ \\
    & & $\EE(t) = \bmureverse(t) + \sigmu(t) + \qmu(t)$ ~(*) \\
    \bottomrule
\end{tabular}
\caption{(Top) Generator set (temporal information atoms) 
	$\SelfI[A|\overline{A}]$ for the set of past, present, and future
	random variables $\RVSet = \left\{ \Past_t, \MeasSymbol_t, \Future_t
	\right\}$. Compare the list of $\SelfI[A|\overline{A}]$ to the areas
	of the information diagram depicted in \cref{fig:iDiagramProcess}(left).
	(After Ref. \cite{Jurg25a}.)
	Note that the forward and reverse entropy rates---$\hmu^+(t)$ and
	$\hmu^+(t)$, respectively---are equal for stationary processes.
	(Bottom) Composite information measures $\hmu^+$,
	$\hmu^-$, and $\EE$; see Refs. \cite{Jame11a,Jame13a,Ara14a}.
	}
\label{tab:InfoAtomsProcess}
\end{table}

However, we glossed over a rather important point---that the information
measures listed above are nominally functions over joint distributions of
infinite, composite variables; specifically, semi-infinite pasts and futures
come into play. And, this presents several technical issues to address.

First, note that the entropy of a sequence of infinitely many random variables
comprising a stationary stochastic process will either diverge or be finite. To
see this, consider the block entropy $\SelfII{\MS{0}{\ell}}$ as a function of
$\ell$. It diverges when the process has positive conditional entropy
$\SelfIcond{\MeasSymbol_0} {\Past_1}$. In contrast, when the conditional
entropy is zero for each measurement, which occurs for periodic signals, the
past and the future are exactly predictable from any measurement and so
$\SelfII{\MS{0}{\infty}}$ contains finite information. These observations apply
to when working with informational quantities over RV blocks.

(Notational reminder: The kind of information---joint, conditional, or
mutual---is conveyed by the arguments to $\SelfII{\cdot}{}$. Thus, the quantity
first listed in the previous paragraph is a joint entropy over the
length-$\ell$ RV block $\MS{0}{\ell}$; whereas the second is a conditional
entropy $\SelfIcond{X}{Y}$. When the vertical bar there is replaced with a colon
one has a mutual information $\SelfImut{X}{Y}$. In addition, when we are
only interested in block length, we do not index off of time $t$, but set $t =
0$ which is allowed by stationarity, and then use $\ell$ to denote block length for a block located at any time $t$.)

Second, we assume that all other information atoms are finite. To show this,
define the \emph{excess entropy} or the total amount of information shared
between the past and the future of the process:
\begin{align}
    \ExcessEntropy = \SelfImut{\Past}{\Future} ~, 
\label{eq:excessEntropy}
\end{align}
where in this case the present $\Present$ is taken as the start of the future
$\Future$. Note that this quantity is equivalent to the sum of the last three
information atoms; see the last two rows in the Table
\ref{tab:InfoAtomsProcess}.

There is a corresponding instantaneous excess entropy at each moment:
\begin{align}
    \ExcessEntropy(t) = \SelfImut{\Past_t}{\Future_t} ~, 
\label{eq:OnlineExcessEntropy}
\end{align}
Note that for stationary, ergodic processes $\ExcessEntropy = \langle
\ExcessEntropy(t) \rangle_t$---the temporal average.

And, it has a corresponding self-information:
\begin{align*}
\ExcessEntropy(\BiInfinity_t)
  = -\log_2 \left( \frac{\Pr(\Past_t,\Future_t)}
  {\Pr(\Past_t) \Pr(\Future_t)}
  \right)
  ~`
\end{align*}
Again, this self-information when setting block delay $\tau = 0$ and time-averaged inside Eq.  \eqref{eq:MutualSelfInformation} becomes the familiar averaged self-excess entropy---the average past-future mutual information
Eq. \eqref{eq:excessEntropy}.

Processes with finite excess entropy $\ExcessEntropy$ (Eq.
\eqref{eq:excessEntropy}) are called \emph{finitary processes}. We will assume
here that all processes under consideration are finitary, and therefore all
other information atoms of the process are also finite, except for the two over
semi-infinite pasts and futures. We can see this by noting the information in
the present $\MeasSymbol_t$ must be finite, as it is a single measurement and
is bounded by $\log_2 | \alphabet |$. So, any information quantity containing
$\MeasSymbol_t$ must be finite---all atoms, except semi-infinite groups, must
be less than or equal to the excess entropy.

Third, each of the finite atomic information quantities may be expressed as the
asymptotic growth rate of a block information quantity
\cite{Jame11a}.
That is to say, the growth rate of an information quantity taken over words of
length $\ell$ as $\ell$ goes to infinity.

For example, consider the total amount of information contained in words of
length $\ell$. As $\ell$ increases, this block entropy scales as:
\begin{align*}
    \SelfII{\MeasSymbol_{0 : \ell }} \sim \ExcessEntropy +
    \ell \hmu , \text{ as } \ell \to \infty,
\end{align*}
where $\hmu$ is the \emph{entropy rate} and is defined:
\begin{align}
    \hmu = \lim_{\ell \to \infty} \frac{
	\SelfII{\MeasSymbol_{0 : \ell }}
	}{\ell} ~.
\label{eq:BlockShannonEntropyRate}
\end{align}
In some applications, this form is also referred to as a \emph{density}.

It can be shown that this asymptotic quantity is equivalent to the entropy in
the present conditioned on the past:
\begin{align}
    \hmu & = \lim_{\ell \to \infty}
	\SelfIcond{\MeasSymbol_0}{ \MeasSymbol_{-\ell:0} }  \nonumber  \\
    & = \SelfIcond{\Present}{\Past}
	~.
\label{eq:ShannonEntropyRate}
\end{align}
This expression gives an operational view of the information measure---the
entropy rate is the amount of new information learned upon a single new
observation of the process. This also allows us to identify the quantity on the
process i-diagram in \cref{fig:iDiagramProcess}(left).

There is also the time-local version:
\begin{align*}
\hmu(t) = \SelfIcond{\MeasSymbol_t}{\Past_t}
\end{align*}
and a self-information version:
\begin{align*}
\hmu(t,\meassymbol) = \selficond{\MeasSymbol_t = \meassymbol}{\Past_t}
   ~.
\end{align*}

The \emph{anticipated information} $\rhomu$ is the complement of the new
information $\hmu$---future information that can be predicted from the past:
\begin{align}
    \rhomu & = \SelfII{\MeasSymbol_0} - \hmu \\
    & = \SelfImut{\Present}{\Past}
    ~.
\label{eq:AnticipatedInfo}
\end{align}

And, there is the time-local version:
\begin{align*}
\rhomu(t) = \SelfImut{\MeasSymbol_t}{\Past_t}
\end{align*}
and a self-information version:
\begin{align*}
\rhomu(t,\meassymbol) = \selfimut{\MeasSymbol_t = \meassymbol}{\Past_t}
   ~.
\end{align*}

Depending on the semi-infinite past or future, the entropy rate and excess
entropy are global measures of the information contained in a process.
Specifically, $\hmu$ is the rate of information created and $\ExcessEntropy$ is
that portion of created information communicated from the past to the future.
The entropy rate is useful because of its immediate interpretability. And, also
as it can be shown to be equivalent to the Kolmogorov-Sinai entropy, a
dynamical invariant of the underlying system (environment).

\subsection{General Atoms}

We can identify other atomic self-information densities for the information
atoms given in Table \ref{tab:InfoAtomsProcess}. These have proved themselves
to be useful. We will describe these quantities and give their self-information
versions and their intuitive meanings here. They play a key role in
interpreting information processing in Sec. \ref{sec:IProcessExamples}'s
example processes. The interested reader should consult Ref.  \cite{Jame11a}
for a deeper discussion. 

\emph{Ephemeral information rate $\rmu$} That information localized to an
isolated variable, but not correlated to its peers. In other words, the rate
can be interpreted as the average information in a single measurement of a
process undetermined by any temporal correlation:
\begin{align*}
  \rmu(t) =  \SelfIcond{\MeasSymbol_t}{ \Past_t,  \Future_t } 
  ~.
\label{eq:ephemeralInfo}
\end{align*}
The self-information version is:
\begin{align*}
  \rmu(t,\meassymbol) =  \selficond{\MeasSymbol_t = x}{ \Past_t,  \Future_t } 
  ~.
\end{align*}

\paragraph*{Binding rate $\bmu$} The rate of increase of the total information
in a block minus the ephemeral entropy. It is also the rate at which a process
stores created information. There are two equivalent quantities, forward
binding rate $\bmuforward$ and reverse binding rate $\bmureverse$. For
stationary processes we always have $\bmuforward = \bmureverse$. The forward
and reverse binding rates can be interpreted as how correlated any given
measurement of a process is with the future and the past, respectively:
\begin{align}
	\bmuforward(t) & = \SelfIcond {\MeasSymbol_t,\Future_t} {\Past_t} ~\text{and} \\ 
	\bmureverse(t) & = \SelfIcond {\MeasSymbol_t,\Past_t} {\Future_t}  
  ~.
\label{eq:bindingInfo}
\end{align}
The two self-information versions are:
\begin{align*}
	\bmuforward(t,\meassymbol) & = \selficond {\MeasSymbol_t = \meassymbol,
	\Future_t }{\Past_t}  ~\text{and}\\ 
	\bmureverse(t,\meassymbol) & = \selficond {\MeasSymbol_t = \meassymbol,
	\Past_t }{\Future_t}  
  ~.
\end{align*}

\paragraph*{Enigmatic rate $\qmu$} The interaction information between any given
measurement of a process and the infinite past and future:
\begin{align}
	\qmu(t) = \SelfI \left[ \MeasSymbol_t ; \Past_t ; \Future_t \right]  
  ~.
\label{eq:EnigmaticInfo}
\end{align}
As this is a multivariate mutual information, it can be negative. The
self-information version is:
\begin{align*}
	\qmu(t,\meassymbol) = \selfi \left[ \MeasSymbol_t = \meassymbol;
	\Past_t ; \Future_t \right]  
  ~.
\end{align*}

\paragraph*{Elusive information $\sigmu$} The amount of information shared
between the past and future that is not communicated through the present. The
elusive information is also sometimes called \emph{state information}, as it
represents the amount of information the observer needs to access beyond
measurement to build a predictive model of the system:
\begin{align}
	\sigmu(t) = \SelfIcond { \Past_t : \Future_t}{ \MeasSymbol_t } 
  ~.
\label{eq:ElusiveInfo}
\end{align}
The self-information version is:
\begin{align*}
	\sigmu(t,\meassymbol) = \selficond { \Past_t , \Future_t}{ \MeasSymbol_t =
	\meassymbol } 
	~.
\end{align*}

\paragraph*{Interpretations} There are a number of useful relationships. Let's
mention just one and recommend the references above for fuller discussion.

The ensemble-averaged forward binding information and ephemeral information
taken together are equal to the ensemble-average entropy rate at each time $t$:
\begin{align*}
  \hmu(t) = \rmu(t) + \bmuforward(t) ~.
\end{align*}
This decomposes the newly-created information per measurement into a part
correlated with the future and a part only correlated with the current
measurement and ``forgotten'' in the next time step. To take one example,
$\bmu(t)$ is the instantaneous rate of storing information internally for use
at some future time.

The ensemble-averaged reverse binding information and enigmatic information
taken together equal the ensemble-average forward anticipation rate at each time $t$:
\begin{align*}
  \rhomu^+(t) = \qmu(t) + \bmureverse(t) ~.
\end{align*}
This decomposes the predictable information per measurement into a part
correlated with the past and a part only correlated with the current
measurement. Similarly for:
\begin{align*}
  \rhomu^-(t) = \qmu(t) + \bmuforward(t) ~.
\end{align*}

These informations are depicted in visual form by the i-diagram in
\cref{fig:iDiagramProcess}(left) and explicitly (minus the ``endcap'' atoms) in
\cref{tab:InfoAtomsProcess}. In principle, it is possible to profile any 
process (or environment) using these atoms to give a full picture of its
informational structure.

\subsection{Dynamic Information Processes}
\label{Sec:DynamicIProcesses}

Implicit in detailing the time-dependence of information atoms is an agent
that moment-by-moment interprets an observation of its environment as
containing useless or useful information---used in predictions, decisions, or
actions. Such time-dependent information measures are quite common---as Sec.
\ref{sec:Background} noted---although perhaps not always called out as such
dynamic quantities.

There are even pitfalls in interpretations. For example, the \emph{transfer
entropy} \cite{Schr00a}---a time-dependent measure of mutual information within
a sliding window, was introduced to detect causal interactions in dynamical
systems. The causation entropy \cite{Sun14a}, also time dependent, was then
introduced to address several of its weakness in conflating conditional
dependence and independence \cite{Jame15a}. Despite such concerns, these
causal-detection measures have come to be widely used in the empirical sciences
for structural inferences. This, at least, attests to the need for such
statistics.

Taken altogether, though, the foregoing lays out the main setting in which such
time-dependent statistics operate. With this, it is now time to be more
explicit about the agent and its functioning.

To set this up, though, we must first address the time evolution of information
measures, which has only been indirectly specified thus far via the shift
operator $\tau$ of Eq. \eqref{eq:ShiftOperator}. Figure
\ref{fig:iDiagramProcess} makes this explicit by contrasting the i-diagrams of
measures at time $t$ (left) and those at the next moment $t+1$
(right)---that is, measures for processes $\MeasSymbol^t$ and
$\MeasSymbol^{t+1}$, respectively.

At time $t$, we have measures of the past, present, and future; viz., $\SelfII{
\Past_{t-1} }$, $\SelfII{\MeasSymbol_t}$, and $\SelfII{\Future_{t+1}}$,
accompanied by the atoms $\bmureverse (t)$, $\rmu(t)$, $\bmuforward(t)$,
$\qmu(t)$, and $\sigmu(t)$. At the next time, we have $\SelfII{ \Past_t }$,
$\SelfII{\MeasSymbol_{t+1}}$, and $\SelfII{\Future_{t+2}}$ and their companions
$\bmureverse (t+1)$, $\rmu(t+1)$, $\bmuforward(t+1)$, $\qmu(t+1)$, and
$\sigmu(t+1)$.

At each time step, $\tau$ transforms the measures at $t$ into those at time
$t+1$, with the former typically splitting and being shared across the latter.
One such atomic transformation is depicted: $\SelfII{\Future_{t+1}}$ maps to
$\bmu^- (t+1)$, $\qmu(t+1)$, and $\sigmu(t+1)$. That is, the time-$t$ forward
block entropy splits and contributes to the time $t+1$ reverse bound
information, enigmatic information, and elusive information.

Reference \cite{Crut24a} provides fuller details and interpretations in other
cases. All in all, though, this gives a view of the dynamic operation of an
information process---extracting, storing, and processing various kinds of
information. Section \ref{sec:IProcessExamples} gives an even more explicit
view of dynamical information processing in when analyzing the information
process time series generated by several examples.

As we are working in the setting of temporal or pointwise information measures, a final cautionary remark is in order. While motivated by elementary
information theory for a given set of random variables, there are important
interpretational differences for information processes.

Concretely, the elementary information theory of two random variables leads
to informations that are positive---entropy, entropy rate, information gain,
mutual information, and so on. It is well-known, however, that when applied to
three random variables the mutual information can become negative. This
corresponds to the previously-mentioned phenomena of conditional dependence
versus conditional independence induced between two random variables by a third
\cite[Prob. 2.25]{Cove06a}. The key point here is that a similar deviation from
intuitions derived from elementary information theory occurs for pointwise or
temporal information measures.

To be specific consider the information processes consisting of residual
information $\rhomu(t)$ and bound information $\bmu(t)$. As conditional mutual
informations they can be negative in the pointwise setting. For $\rhomu(t)$,
which is a simple mutual information, negativity appears when its associated
self-information argument $\Pr\left(\Past(t), \MeasSymbol_t\right) /
\Pr\left(\Past(t)\right) \Pr\left(\MeasSymbol_t\right) < 1$ (or when
$\Pr\left(\Past(t), \MeasSymbol_t\right) < \Pr\left(\Past(t)\right)
\Pr\left(\MeasSymbol_t\right)$). This occurs in the pointwise setting when the
two observations---$\MeasSymbol_t$ and $\Past(t)$---occur less often than one
would expect if the process were independent, identically distributed.
Whereas, $\rhomu(t) > 0$ implies the past and present occur more frequently
than one expects if they were independent. Section \ref{sec:IProcessExamples}
provides numerous examples of this and the negativity of other temporal
information measures that one would, without warning, assume positive. The
practical result is a need to retool intuitions from elementary information
theory when interpreting temporal informations.

\section{Cognitive Agents}
\label{sec:Agent}

The following introduces a more explicit notion of (i) an agent interacting with
(ii) an environment whose behavior the agent monitors by making sequential
observations, interpreting the latter to, perhaps, make decisions in the
service of future interactions.

We refer to the sequential measurement and interpretation as \emph{cognition}
and the observers as \emph{cognitive agents}. Granted this invokes a rather
literal notion of cognition---one that falls short of that in animals say.
That said, the goal in the following is to lay out the basic dynamical and
informational foundations for cognitive agents of any stripe. The result of
this is to demonstrate how a cognitive agent is a transducer that maps an input
stream of observations to other informational variables specified by an agent's
organization and design---those environmental patterns to which it is
sensitive. To do this, we frame the notion of a cognitive agent in terms of
computational mechanics---the information theory of structured stochastic
processes.


Moreover, there is a practical motivation. Working with processes, as we have
above---nominally, infinite sets of infinite sequences and their
probabilities---is cumbersome, at best. Typically, we do not care to estimate
informations over distributions of infinite pasts and futures. Nor are agents
infinitely resourced for doing so.

So, we turn to finitely-specified representations. That is, we turn to models
specified by \emph{computational mechanics}---called \eMs and \eTs---which have
especially desirable properties \cite{Shal98a,Barn13a} for determining and
estimating information processing, both asymptotic properties and, as we show
now, moment-by-moment, online statistics. Through their optimal representations
of stochastic processes and communication channels, they allow us and cognitive
agents, for that matter, to work with finite objects rather than semi-infinite
sets and sequences. Helpfully, this often leads to exact, closed-form
expressions for many properties of interest \cite{Crut13a}. Again, the
following assumes that a cognitive agent's internal model, unless otherwise
stated, effectively employs its minimal optimal predictor---the observed
process' \eM.


\subsection{Predictive States}
\label{sec:Predictive States}

We formalize a prediction as a distribution $\Pr(\Future|\past)$ over futures
$\{\Future\}$ with knowledge of a given past $\past$. We wish to construct a
minimal model that produces optimal predictions for a stochastic process
$\Process$. Computational mechanics \cite{Shal98a} solved this problem in the
form of the \emph{\eM}---a model whose states are the classes defined by an
equivalence relation $\past \sim \past '$ that groups all pasts giving rise to
the same prediction $\Pr(\Future|\past)$. These classes are called the
\emph{causal states}.

Computational mechanics also provides an analogous solution for optimal
predictions of one process, call it $X$, about another, call it $Y$
\cite{Barn13a}. The minimal, optimal model for the conditional input-output
process or communication channel---denoted $Y|X$---is called an \eT. Since
the following makes minimal calls on the latter, here we concentrate on
\eMs. Sequels rely more heavily on transducers.

\begin{definition}
\label{Def:CausalStates}
The \emph{causal states of a process} are the members of the range of the
function:
\begin{align*}
  \EquiFunction{\past} = 
  \Big\{ \pastprime \mid  & 
  \Pr \left( \Future = \future | \Past = \past \right)  \\
  = & \Pr \left( \Future =  \future | \Past = \pastprime \right) \\
  & \text{for all} \:\: \past \in \Past, \pastprime \in \Past \Big\}
\end{align*}
that maps from pasts to sets of pasts. The set of causal states is denoted
$\CausalStateSet$, with corresponding random variable $\CausalState$ and
realizations $\causalstate \in \CausalStateSet$. 
\end{definition}

The causal states can be both empirically and measure-theoretically grounded.
First, from the (past,future) pairs in a realization $\past \future$ construct
an empirical conditional distribution:
\begin{align}
  \widehat{P}_{x_{1}\dots x_{n}|x_{-k+1}\dots x_0} =
  \frac{C_{x_{-k+1}\dots x_{n}}}{C_{x_{-k+1}\dots x_0}}
  ~,
\label{eq:Likelihood}
\end{align}
where $C_{w}$ is the number of times the word $w$ appears in the sequence $x_1
\dots x_L$. Given sufficient data, it would be desirable to take pasts of
arbitrary length and converge towards a prediction conditioned on the
\emph{infinite} past:
\begin{align} 
  {P}_{x_{1}\dots x_{n}|\overleftarrow{x}}
  = \lim_{k\rightarrow\infty}
  \widehat{P}_{x_{1}\dots x_{n}|x_{-k}\dots x_0}
\label{eq:convergence}
\end{align}
with the infinite sequence $\overleftarrow{x} = (\dots,x_{-1},x_0)$ of
observations stretching into the past.

Second, formally, the conditional predictions ${P}_{x_{1}\dots
x_{n}|\overleftarrow{x}}$ for all forecast lengths $n$ together describe a
\emph{probability measure} over future sequences $\overrightarrow{x} =
(x_1,x_2,\dots)$.  In short, a causal state is a measure over future sequences
conditioned on past sequences: $\Pr_\mu(\Future|\Past = \ldots \meassymbol_{-2}
\meassymbol_{-1} \meassymbol_0 )$. We denote it simply $P_{\overleftarrow{x}}$.

Reference \cite{Loom21b}
established a number of useful convergence and
topological properties of \eM casual states, including:
\begin{enumerate}
      \setlength{\topsep}{-4pt}
      \setlength{\itemsep}{-4pt}
      \setlength{\parsep}{-4pt}
\item In the language of measures, as an agent collects more observations,
	$P_{\overleftarrow{x}}$ converges \emph{in distribution}, with respect
	to the \emph{product topology} of the space of sequences
	$\mathcal{X}^\mathbb{N}$ \cite{Loom21b}.
\item This also holds for convergence over word distributions
	$\Pr(\meassymbol_{1:n})$.
\item A conditional measure is a ratio of likelihoods, as in Eq.
	(\ref{eq:Likelihood}).
\item For all measures $\mu$ on $\mathcal{X}^\mathbb{Z}$ and $\mu$-almost every
	past $\Past \in \mathcal{X}^\mathbb{Z}$, the measures $\eta_\ell[\Past]$
	defined by:
\begin{align*}
\eta_\ell[\Past](U_{0,\omega})
	= \Pr_\mu (\omega | \meassymbol_{-\ell+1} \ldots \meassymbol_0)
\end{align*}
	converge in distribution to the measure $\epsilon[\Past]$ : $\eta_\ell[\Past] \to \epsilon[\Past]$, as $\ell \to \infty$.
\item $\epsilon[\Past]$ is the \emph{predictive state} of $\Past$ and the
	function that maps to measures over future sequences---$\epsilon[\Past] \to
	\mathbb{M} (\mathcal{X}^\mathbb{Z})$---is the \emph{prediction mapping}.
\end{enumerate}

Anticipating future uses, the \emph{causal states of an input-output
process}---or \emph{communication channel}---are similarly defined by
equivalence classes of channel input sequences and output sequences. The
corresponding mapping from an input process to an output process is called the
\emph{$\epsilon$-transducer}. The following only considers \eMs for processes,
leaving the \eT development for channels to a sequel in which they play a
pivotal role.

In this way, causal states partition the space of all pasts into sets that are
predictively equivalent. The causal state set $\CausalStateSet$ may be finite,
fractal, or continuous, depending on $\Process$'s properties
\cite{Crut91b}.

The following focuses on processes with finite causal-state sets: $|
\CausalStateSet | < \infty$. Sequels remove this restriction.

\subsection{Causal-State Processes}
\label{sec:CSProcess}

Given a stochastic process $\Process$, consider a realization:
\begin{align*}
\ldots, \msym_0, \msym_1, \msym_2, \ldots
  ~.  
\end{align*}
From this, form the associated time series of semi-infinite histories:
\begin{align*}
\ldots, \past_0, \past_1, \past_2, \ldots
  ~,
\end{align*}
where $\past_t = \ldots \msym_{t-2}, \msym_{t-1}, \msym_t$.

Now, we use \emph{causal-state filtering} to obtain the causal-state
realization---mapping from a time series of observations to a process' causal
states:
\begin{align*}
& \ldots, \epsilon(\past_0), \epsilon(\past_1), \epsilon(\past_2), \ldots \\
& \ldots, \causalstate_0, \causalstate_1, \causalstate_2, \ldots
  ~;
\end{align*}
that is, $\causalstate_t = \epsilon(\past_t)$.

We define two kinds of causal-state behavior of interest:
\begin{itemize}
      \setlength{\topsep}{-4pt}
      \setlength{\itemsep}{-4pt}
      \setlength{\parsep}{-4pt}
\item The \emph{causal-state process} $\CausalState = \{\CausalState_t: t \in
	\mathbb{Z}\}$ is the \emph{temporal sequence of causal states} $\CausalState_t
	= \epsilon[\tau^t \past]$; and
\item \emph{Recurrent causal-state process} is that series of states after the
	agent is synchronized to the environment.
\end{itemize}
We work with the latter, though sequels lay out several of the challenges of
working with the transient, nonstationary statistics of the former.

\subsection{\texorpdfstring{$\epsilon$}{e}-Machines}
\label{sec:eMachines}

The dynamic over the casual states is inherited from the shift operator
$\tau$ on the process. State-to-state transitions occur when observing a new
symbol $\meassymbol$, which is appended to the observed history: $\past \to
\past \meassymbol$. The causal state transition is therefore from
$\EquiFunction{\past} = \causalstate_i \to \EquiFunction{\past \meassymbol} =
\causalstate_j$, and occurs with probability $\Pr \left( \meassymbol =
\meassymbol \mid \CausalState = \causalstate_i \right)$. 

\EMs are guaranteed to be optimally predictive because knowledge of what
causal state a process is in at any time is equivalent to knowledge of the
entire past: $\Pr \left( \FutureSmashed \mid \CausalState \right) = \Pr \left(
\FutureSmashed \mid \PastSmashed \right)$. They are also Markovian in that they
render the past and future statistically independent: $\Pr \left( \PastSmashed,
\FutureSmashed \mid \CausalState \right) = \Pr \left( \PastSmashed \mid
\CausalState \right) \Pr \left( \FutureSmashed \mid \CausalState \right)$. We
call these properties together \emph{causal shielding}. \EMs also have a
property called \emph{unifilarity}, which means that knowledge of the current
causal state and the next symbol is sufficient to determine the next
state. That is to say, $\SelfIcond{\CausalState_{t+1}}{\MeasSymbol_t,
\CausalState_t} = 0$.

\begin{definition}
\label{def:eMachine}
The \emph{\eM} $\eMachine$ of a finitary process consists of:
    \begin{enumerate}
      \setlength{\topsep}{-4pt}
      \setlength{\itemsep}{-4pt}
      \setlength{\parsep}{-4pt}
        \item Finite \emph{alphabet} $\alphabet$ of $\numSyms$ symbols
			$\meassymbol \in \alphabet$;
        \item \emph{Casual state} set $\CausalStateSet = \left\{ \causalstate_1,
        	\causalstate_2, \dots \right\}$ that consists of transient states
			whose probabilities vanish and \emph{recurrent} states
			$\tilde{\causalstate}$ whose probability converges to a constant
			$\Pr(\tilde{\causalstate}) > 0$. And;
        \item \emph{Causal dynamic}---set of $\numSyms$ (possibly infinite
		dimension)
			symbol-labeled transition matrices $T^{(\meassymbol)}$, $\meassymbol \in
			\alphabet$: $T^{(\meassymbol)}_{ij} = \Pr \left( \causalstate_j,
			\meassymbol \mid \causalstate_i \right)$.
    \end{enumerate}
\end{definition}

This defines its own class of optimal representations for a wide range of
stochastic processes. That said, we can draw several parallels with existing
classes of process representations. The definition here identifies an \eM as a
\emph{hidden Markov model} (HMM). Not all HMMs are \eMs. However, the \eMs here
are. An \eM may be graphically shown as an HMM with a directed graph
where the causal states are depicted by vertices and transitions between them by
directed edges labeled with the symbol emitted on transition followed by the
probability of transition; e.g., $\meassymbol : \Pr\left(\meassymbol \right)$.
The time indexing is as follows: if at time $t$, an \eM is in state
$\CausalState_t$, it emits symbol $\meassymbol_t$ and transitions forward to
the next state $\CausalState_{t+1}$. Notice that due to unifilarity, there is
at most one transition from each causal state per symbol. 

\subsection{Agent-Environment Synchronization}
\label{sec:AESync}

\begin{proposition}
Given the causal state at time $t-1$, the causal state at time $t$ is
independent of the causal states at earlier times $t < t-1$.
\end{proposition}

\begin{proof}
See Lemma 6 (\EMs are Markovian) in Ref. \cite{Shal98a}.
\end{proof}

That is, the causal-state process is order-$1$ Markov:
\begin{align*}
\Pr(\causalstate_t | \ldots \causalstate_{t-2} \causalstate_{t-1}) =
\Pr(\causalstate_t | \causalstate_{t-1})
  ~.
\end{align*}

\newcommand{\wcs}{\widetilde{\causalstate}}

Markovity induced by the causal states makes it particularly straightforward to
give the recurrent causal-state process directly in terms of its word
distributions:
\begin{align*}
\Pr(\wcs_{0:n}) = (\pi_0)_{\wcs_0} \prod_{t=0}^n T_{\wcs_t,\wcs_{t+1}}
  ~,
\end{align*}
where $\wcs_t$ is a recurrent causal state, $T$ is the state-transition
operator, and $(\pi_0)_{\wcs_0}$ is the initial causal state distribution. (The
tildes identify recurrent states.)

The simplicity of the casual-state process demonstrates concretely why the
causal states are so useful. For one, the summary they give of the past $\Past$
is sorely needed to make various potentially-infinite calculations tractable.

Leveraging this, we define \emph{agent-environment synchronization} when the
agent is fully tracking the environment. First, define a synchronizing word
$\widehat{w}$: $\Pr(\wcs|\widehat{w}) = 1$, for one $\wcs \in \CausalStateSet$,
and where $\widehat{w}$ is one of the words $\mathcal{L}(\Process)$ generated by the process: $\widehat{w} \in \mathcal{L}(\Process)$. The
informational definition of synchronization is that the agent exactly knows the
current state: $\SelfIcond{\causalstate_t}{\widehat{w}} = 0$.

\subsection{Agent Operation}
\label{sec:Operations}

There are two principle modes of operation for \eMs: Recognizing sequences of
environment behaviors versus generating environment control signals:
\begin{itemize}
      \setlength{\topsep}{-4pt}
      \setlength{\itemsep}{-4pt}
      \setlength{\parsep}{-4pt}
\item \emph{Recognition mode}: This addresses the question, is a given word $w
	\in \alphabet^*$ in the set (or language) recognized by the \eM? Is $w \in
	\mathcal{L}(M)$?
\item \emph{Generation mode}: The mode concerns what words $w \in \alphabet^*$
	are emitted by an \eM.
\end{itemize}
As noted below, these modes of operation come into play when an agent encounters
stimuli disallowed by its internal model. That is, when it resets its state in
response.

Before use, an \eM must be properly initialized. In either recognition or
generation mode, the procedure to initialize an \eM sets the initial
causal-state distribution: Set the current state to \eM's unique start state
$\CausalState = \causalstate_0$ and set the current state distribution $\pi_t$
to have unity probability on $\causalstate_0$: $\pi_0 = (1, 0, 0, \ldots)$.

Given an \eM, its asymptotic causal-state distribution $\widehat{\pi}$ is the
left-eigenvector of the causal-state transition operator $T = \sum_{\meassymbol
\in \alphabet} T^{(\meassymbol)}$:
\begin{align*}
\widehat{\pi} = \widehat{\pi} T
  ~,
\end{align*}
normalized in probability: $\sum_{\causalstate \in \CausalStateSet}
\Pr(\causalstate) = 1$.

In each mode, the \eM operates step-by-step, sequentially reading the symbols
in a word, to accept or reject the word as follows.
\begin{itemize}
\item \emph{Acceptance}: Reach the input word's last symbol while in recurrent causal state $\causalstate_t \in \CausalStateSet$.

Sequentially reading symbols from the input, update the current state to that
reached by transition labeled by $\msym_t$ and update current state
distribution:
\begin{align*}
\pi_{t+1} (\msym) = \frac{\pi_t T^{(\msym)} }{ \sum_\msym \pi_t T^{(\msym)} }
  ~.
\end{align*}
\item \emph{Rejection}: While reading symbols from the input the agent
encounters a disallowed symbol; that is, from the current state there is no
transition labeled with that symbol.

Then, reset the current state to the agent's unique start state
$\causalstate_0$ and reset the current-state distribution $\pi_t$ to unit
probability on the start state $\causalstate_0$:
\begin{align*}
\pi_{t+1} = (1, 0, 0, 0, \ldots)
  ~.
\end{align*}

For generation mode the \eM operates according to a companion \emph{emission}
procedure that simply follows the allowed transitions emitting the symbols
labeling each.
\end{itemize}

The probability of generating word $w = x_0 x_1 \ldots x_{\ell-1}$ given start
distribution $\pi_0$ is:
\begin{align}
\Pr(w) & = \pi_0 \prod_{i=0}^{\ell-1} T^{(x_i)} \nonumber \\
  & = \pi_0 T^{(\omega)}
  ~.
\label{eq:WordDistribution}
\end{align}
If an \eM starts with $\widehat{\pi}$, then $\pi_0 = (1, 0, ...)$. Recall that
the start state corresponds to the asymptotic state distribution; a condition
representing unbiased information about the environment's state.

\begin{proposition}
If an \eM used $\pi_0 = \widehat{\pi}$ to generate word $w$, then $\Pr(w)$ is
$w$'s asymptotic stationary probability.
\end{proposition}

\begin{proof}
By elementary Markov chain theory \cite{Fell70a}.
\end{proof}

\begin{definition}
If $M$ is an \eM, its \emph{output process} $\Process(M)$ consists of the word
distributions given by Eq. (\ref{eq:WordDistribution}) above.
\end{definition}

\begin{proposition}
If starting with $\widehat{\pi}$, an \eM's output process is stationary.
\end{proposition}

\begin{proof}
Elementary Markov chain theory \cite{Fell70a} shows that word distributions are
time-shift invariant.
\end{proof}

\subsection{Intrinsic Information Processing}
\label{sec:InfoStore}

There are infinitely many possible optimal predictive models---if thinking of
state-based models, imagine adding redundant or even useless states.
Helpfully, computational mechanics established that the \eM is a process'
minimal and unique predictive model in the sense that the amount of information
stored by the causal states is smaller than (or equal) to any other possible
\emph{prescient} (equally predictive) rival. We quantify this via the Shannon
entropy of the causal-state distribution---this is the \emph{statistical
complexity}: $\Cmu = \SelfII{\CausalState} = \SelfII{\widehat{\pi}}$.

With this, there are three basic informational invariants that describe a
stationary process: $\hmu$ is a process' rate of information creation, $\Cmu$
is the amount of history a process remembers from its past $\Past$---the
historical information stored in the causal states $\CausalStateSet$; and $\EE$
is the amount of that stored information communicated to the future $\Future$.



Generally, the following refers to an environment's \emph{inherent information}
if an optimally-predicting agent uses the environmental process' \eM and is
synchronized.

If the agent uses a model other than the \eM, then the information available is
relative to that model and to being synchronized or not. Hence, this highlights
the distinction between ``subjective'' self-information---using any model---and
``inherent'' self-information---using the \eM. Section \ref{sec:Meaninglessness}
below illustrates the consequences of an agent using an incorrect environment
model.

Finally, we generically refer to the time series of temporal information
measures and the causal-state and prediction processes collectively as
\emph{information processes}. These are the time series generated during an
agent's online operation that allow us to monitor and diagnose the informational
processing that it performs while interacting with an environment.

\section{Intrinsic Semantics of Information Processes}
\label{sec:ProcSemantics}

Generally, one concludes that an \eM\ captures the patterns in a stochastic
process \cite{Shal98a}. However, focusing on the global, time-asymptotic view
side-steps a key question that Fig. \ref{fig:MSemantics} highlights:
\begin{quote}
What does a particular measurement mean?
\end{quote}
To address this, the following reviews and then applies the notion introduced
some time ago in Ref. \cite{Crut91b} of \emph{measurement semantics}. This
refers to the (informational) meaning revealed to an agent via measurements of
its environment with respect to its internal model of that environment.

Said more plainly, we assume an agent has a model in its ``head'' that it uses
to interpret the incoming sequence of measurements. The following explores the
relationship between incoming ``raw'' observations and how an agent deploys a
(good or bad) model of the observed process' generator to interpret them.

It might use this model to ``understand'' patterns in the process or simply may
use it to predict future observations. We start with the latter first.

Two points to clarify first. The following assumes an agent has a model of its
environment. As such, the following does not discuss how an agent obtains its
model. This certainly brings up important, but for now peripheral, issues; for
example, the issues of statistical estimation, inference, and overfitting. See,
for example, Ref. \cite{Stre13a} and references therein. Beyond inference,
though, there are many competing ways an agent can have a model of its
environment. It can, for example, simply sample a given model distribution---a
topic for a separate development.

In addition, the following assumes the agent uses the incoming observations and
its model to synchronize to the environment's state. The way in which it does this can be quite challenging. It is its own topic; see
Ref. \cite{Trav14a}.

\subsection{Prediction Semantics}
\label{sec:PredictionSemantics}

One interpretational setting is to use past observations $\past_t$ up to time
$t$ to predict the future, say, $\msym_{t+1}$ or even the whole future
$\future_t = \msym_{t+1}, \msym_{t+2}, \ldots$. But to what end? Specifically,
even mere prediction begs asking, What is the meaning of the particular
measurement $\msym_{t+1} \in \alphabet$?

Harking back to Sec. \ref{sec:InfoTh}, we have the following.

\begin{definition}
Given a particular past $\past_t \in \Past$, recall that Shannon defined the
amount of \emph{self-information} an agent gains in observing $\meassymbol \in
\alphabet$ to be \cite{Shan48a}:
\begin{align}
-\log_2 \Pr(\MeasSymbol_t = \meassymbol)
  ~.
\end{align}
\end{definition}

This is, in short, an agent's degree of surprise on observing $\msym_{t+1}$. As
Shannon formulated it and as noted above, the degree of surprise is the most
basic concept in information theory. Beyond the intuition motivating
self-information, Ref. \cite{Shan48a} did not specify exactly how one obtains
the probability $\Pr(\MeasSymbol_t = \msym$). It is assumed to be available to
an agent ... (or, in Shannon's setting, available to the communications channel
analyst.)

And so, we must refer to this as the \emph{subjective self-information} since,
nothing else said, it does not specify how the probability is determined.
(This is in stark contrast to the comments at the beginning, which assumed
probabilities were given.) Specifically, any interpretive model could be used
to develop $\Pr(\MeasSymbol_t = \msym)$ and the self-information would change
accordingly. The need to address this issue is now clear. The following
assembles the required components.

We can specialize the above definition of \emph{subjective} self-information by
constraining the model used, rather than it being arbitrary. (Again, that
arbitrariness led us to the label ``subjective''.) Having identified this
relativity, when using information theory one is chastened to specify at the
outset the model class of process generators with which one works.

Heeding this, one can then specify a model class, such as Markov chains or
hidden semi-Markov processes, as appropriate. However, there is still model
choice within a given class and so follow-on informational interpretations are
subjective. Computational mechanics' introduction of a process' \eM solves this
problem, grounding interpretation in a process' minimal and unique
representation that comes from the process itself. Seemingly simple, this is
one of computational mechanics' broader contributions as it removes unnecessary
subjectivity in applying information theory to a given stochastic process..

Computational mechanics adopts a companion view of a communication channel that
refines Shannon's notion of the channel \emph{receiver} by making explicit the
observing entity---the agent; in fact, one that is optimally predicting. In
short, for objective or intrinsic information measures, an agent employs a
process' \eM as its internal model.

A key consequence of its optimality is that using the \eM as the agent's
internal model grounds much of information theory that is unspecified or, as
noted, subjective.

\begin{definition}
Given a stationary process $\Process$ and an agent synchronized to it that uses
$\Process$'s \eM $M = \{ \alphabet, \CausalStateSet, \CausalTransitionSet,
\mu_0\}$, the \emph{intrinsic self-information} in observing $\msym$ is:
\begin{align}
-\log_2 & \Pr(\MeasSymbol_t = \msym) \nonumber \\
	& = -\log_2 \Pr(\past_t)
	\Pr \left( \causalstate_t \rightarrow_{\msym} \causalstate_{t+1} | \past_t
	\right)
  ~.
\end{align}
Here, the self-probability $\Pr \left( \causalstate_t \rightarrow_{\msym}
\causalstate_{t+1} \right)$ is calculated from the \eM, as follows: (i) Note
the probability $\Pr(\causalstate_t|\past_t)$ of being in state $\causalstate$
having observed $\past_t$. Applying the $\epsilon(\cdot)$, this is unity.
And, (ii) multiply that by the (transition) probability
$\LabelCausalTransition$ of taking the $\msym$-labeled transition. That is:
\begin{align}
\Pr(\MeasSymbol_t = \msym)
	& =  \Pr(\past_t) T^{(\msym_t)}_{\epsilon(\past_t) \to_{\msym} \CausalStatePrime} \\
	& =  \Pr(\past_t) T^{(\msym)}_{\causalstate_t \to_{\msym}\CausalStatePrime}
  ~.
\end{align}
Note the dependence on the process' asymptotic invariant measure $\mu$,
implicitly in using $\Pr(\past_t)$.
\end{definition}

Again, using the \eM is a natural choice given that it is the process' minimal
optimal predictor. And so, follow-on results inherit a number of desirable
properties---such as, actually describing the given process and facilitating
efficient estimation of informational properties. In addition, Ref.
\cite{Shal98a} established that the causal-state process---the temporal sequence
of causal states the agent visits---is an order-$1$ Markov process. This greatly
simplifies much technical development since each causal state summarizes in a
single RV $\causalstate_t= \epsilon(\past_t)$ all of the prediction-relevant
information from the past.

To emphasize, the present setting and the following assume the agent is
\emph{synchronized} to the incoming process---it knows which causal state
($\causalstate$ above) the process generator is in. For stochastic processes
generated using finite-memory there are two
kinds \cite{Trav14a}:
(i) exact synchronization
in which the agent knows the state after a finite series of
observations and (ii) asymptotic synchronization where an infinite series of
observations is required and the agent converges (exponentially fast) to
synchronization.

Agent-environment synchronization is an important property for the results here.
A number of consequences arise when removing the synchronization requirement,
especially when addressing the potentially transient, nonstationary epoch prior
to synchronization. Fortunately, the conditions of exact and asymptotic
synchronization can be substantially broadened to include processes generated by
\emph{path-mergeable} HMMs \cite{Trav14a}.

\subsection{Online Prediction: An Example}
\label{sec:OnlinePrediction}

\setcounter{MaxMatrixCols}{13}

The following assumes the agent uses the process' minimal
optimally-predictive model---the \eM.

To illustrate the operation of prediction, let's consider how an agent
sequentially monitors (optimally) observations in a realization of the Even
Process:
\begin{align*}
\begin{matrix}
	t   & = 0 & 1 & 2 & 3 & 4 & 5 & 6 & 7 & 8 & 9 & 10 & 11 \\
	s_t & = 0 & 1 & 1 & 1 & 1 & 0 & 1 & 1 & 0 & 1 &  1 &  1 \\
\end{matrix}
\end{align*}
Now, at time $t = 11$, the agent measures $s_{11} = 1$. (For reference, see the
Even Process' \eM in Fig. \ref{fig:ExampleProcesses}(d).)

How much information does $\MeasSymbol_{11} = 1$ convey to the agent? We apply
the self-information conditioned on the preceding eleven observations:
\begin{align}
\SelfIcond{ s_{11} }{ s_{10} = 1 , s_9 = 1, \ldots} & \approx \hmu \\
  & \approx 0.585 ~\text{bits}
  ~.
\end{align}
This is the degree of the observer's surprise (unpredictability).

\subsection{Measurement Semantics}
\label{sec:MeasuSemantics}

Crucially, this indicates uncertainty but does not determine what the event
$s_{11} = 1$ means to the observer!

We have one event---$s_{11} = 1$, but there are two contexts or interpretive
levels in which to interpret it. The first is how much it is
unanticipated---the semantics of prediction. The second, concurrent setting is
how it updates the agent's actual anticipation. That is, to what model state
$\causalstate_t$ does the observation take the observer? The meaning to the agent comes from a
tension between representation of the same event $s_{11} = 1$, but at different
levels or contexts. Here, we have:
\begin{enumerate}
\item Level 1 is the data stream and the event is a measurement;
\item Level 2 is the agent and the event updates its model.
\end{enumerate}
With this in mind, we have the following definition of an agent's observational
semantics.

\begin{definition}
The \emph{degree of meaning} of observing $\msym \in \alphabet$:
\begin{align}
\Theta (\msym) = - \log_2 \Pr (\to_\msym \causalstate) 
  ~,
\end{align}
where the arrow notation signifies $\causalstate \in \CausalStateSet$ is the
causal state \emph{to which} $\msym$ brings the agent.
\end{definition}

Said differently, this broadens our understanding of causal states---they are
\emph{contexts of interpretation} within which observed values appear. The set
$\CausalStateSet$, thus, is agent's palette of semantic interpretations of
environment behaviors.

In the preceding example of the Even Process, the state selected is $D$ which
has asymptotic probability $\Pr(\CausalState = D) = 2/3)$ and so:
\begin{align*}
\Theta (\msym_{11} = 1) & = -\log_2 \Pr(\causalstate_{11} = D) \\
  & = -\log_2 (2/3) \\
  & \approx 0.58496 ~\text{bits}
  ~,
\end{align*}
of semantic information.

Moreover:

\begin{definition}
The \emph{meaning content} is the state $\causalstate$ selected from the palette
of anticipations---the model's state set $\CausalStateSet$.
\end{definition}

Again, in the preceding example of the Even Process, $\msym_{11} = 1$ selects
state $D$. And this means that the $1$---compared to other $1$s that may be
emitted---is the second $1$ in the symbol-pairing---$0(11)^n0$---that the Even
Process emits.

And, this highlights the origin of the measurement's meaning \emph{within} the
set of the agent's causal states. That is, the meaning depends on the entire set
of interpretive contexts---within the causal state set $\CausalStateSet$.
That is to say that the interpretive context of an incoming observation is the
causal state the agent is currently in. For one, the causal state has an
associate prediction---its future morph. For another, the current causal state
determines which next state the agent will move to.
Both predicting the next observation and the next state are interpretations of
what the current observation means to the agent.

\subsection{Meaninglessness}
\label{sec:Meaninglessness}

It is important to highlight that this measurement semantics allows for measurements to be meaningless. Here, are several cases.

\begin{enumerate}
\item Given how the \eM operates---recall the acceptance and rejection modes
described above---one specifies an initial state distribution. And, under one
operation mode, the \eM puts all of the probability on the start state. That
is, $\pi_t$ is set to $\pi_0 = (1, 0,0, 0, \ldots)$. Or, in other words, $\Pr
(\causalstate_0) = 1$. And, the symbol observed is $\msym_0 = \lambda$. The
degree of meaning, then, is:
\begin{align*}
\Theta (\msym_0 = \lambda) & = - \log_2 \Pr(\causalstate_0) \\
  & = -\log_2 1 \\
  & = 0
  ~.
\end{align*}
\item Disallowed transition: Upon seeing a disallowed symbol, the agent resets
its \eM to the start state---the state of total ignorance. As just noted, this
means the state, being the start state, has full probability and so is
meaningless. This is intuitive: The disallowed transition or forbidden symbol
observed is meaningless. The model has no interpretation to give---nothing to
say: $\Theta = 0$.
\end{enumerate}

Thus, by the above, the start state is meaningless. Indeed, having seen nothing,
all futures are possible.

Somewhat glibly, meaningless measurements are (very!) informative: 
\begin{align*}
  -\log_2 \Pr(\causalstate_t \to_{\msym_t}\CausalState_0) & = -\log_2 0 \\
  & = \infty
  ~.
\end{align*}
A meaningless measurement is informative in the sense that a zero probability
observation is infinitely surprising. And, moreover, a meaningless measurement
implies the existence of (at least one) zero-probability measurement.

Thus, even if the agent has an incorrect model of a process there is a
semantics of its (likely misleading) interpretations of the environment
behaviors. In such circumstances, there can be many symbols and transitions
that the agent interprets as disallowed. No matter, the above theory properly
describes the semantics according to the agent. Section \ref{sec:GMInterpEP}
illustrates the circumstance when an agent uses as its internal model the Even
Process \eM to (incorrectly) interpret an environment obeying the Golden Mean
Process.

Finally, how this semantic theory applies to hierarchically-structured
processes---such as an environment at the onset of chaos or the particle
interactions in cellular automata distributed computation---are the subjects of
a sequel.

\subsection{Total Semantic Information}
\label{sec:TotalSemanticInfo}

It turns that the total amount of semantic information in a process is related
to its information storage.

\begin{theorem}
A process' total average semantic information is its statistical complexity:
\begin{align*}
\langle \Theta (\msym) \rangle = \Cmu
  ~.
\end{align*}
\end{theorem}

Proof:
\begin{align*}
\langle \Theta (\msym) \rangle & =
  \sum_{\causalstate \in \CausalStateSet} \Pr(\causalstate) \Theta (\msym) \\
  & = - \sum \Pr(\causalstate) \log_2 \Pr(\causalstate) \\
  & = \SelfII{\CausalState} \\
  & = \Cmu
  ~.
\end{align*}

In other words, the average amount of meaning is the statistical complexity
$\Cmu$. This gives an insightful connection between a process' internal
structure and the semantics an observer attributes to observations.

Tables \ref{tab:InfoSemanticAnalysisBiasedCoin},
\ref{tab:InfoSemanticAnalysisPeriod2}, \ref{tab:InfoSemanticAnalysisGoldenMean},
and
\ref{tab:InfoSemanticAnalysisEven}
present an agent's informational and semantic analysis of 
Sec. \ref{sec:IProcessExamples}'s example processes, respectively.

\subsection{Ergodic Theory Redux }
\label{sec:ErgodicPrediction}

The following establishes that the information processes generated by a
(finitary) agent are well-behaved in the sense that the information processes
resulting from observing a stationary and ergodic environment are well-behaved.
Note that such properties are an aid to further, downstream processing that an
agent is tasked to perform.

This holds for any class of generated information process---whether entropy
rate or excess entropy or others based on the various self-informations or
densities (self-informations) developed above. Given this, we then develop
analogous results for when an agent observes an ergodic process in its
environment.

\begin{figure*}
(a) \includegraphics[width=.4\columnwidth]{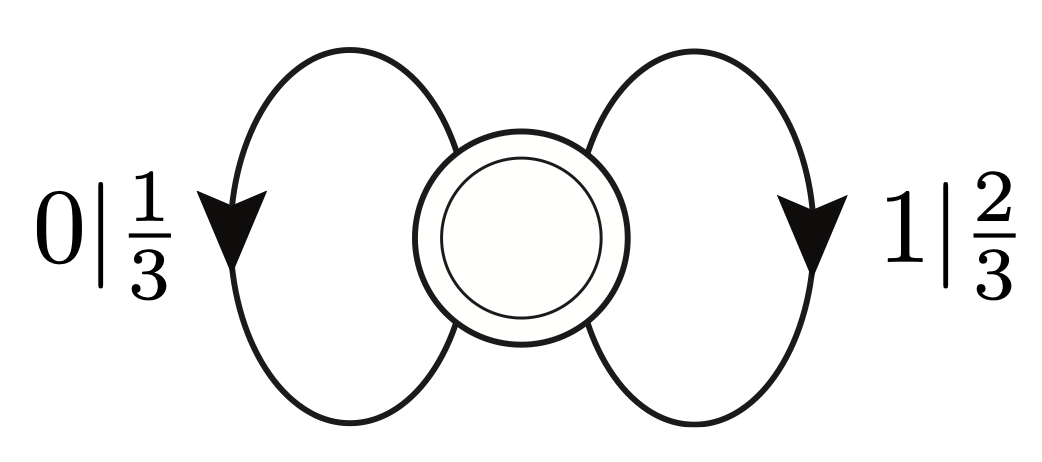} (b)
\includegraphics[width=.4\columnwidth]{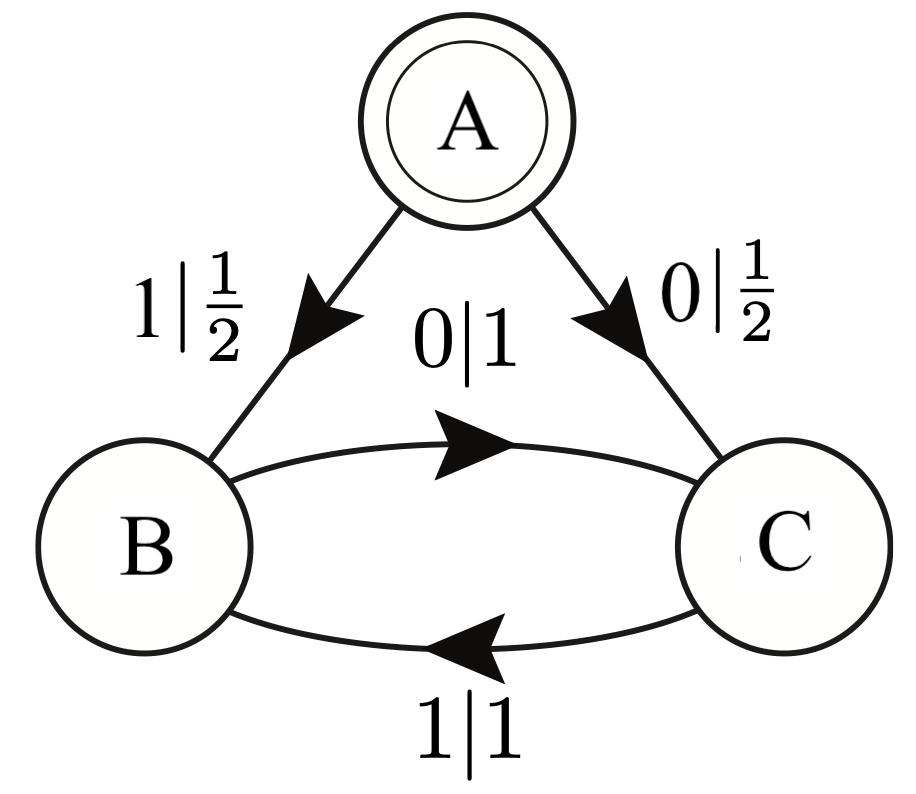} (c)
\includegraphics[width=.4\columnwidth]{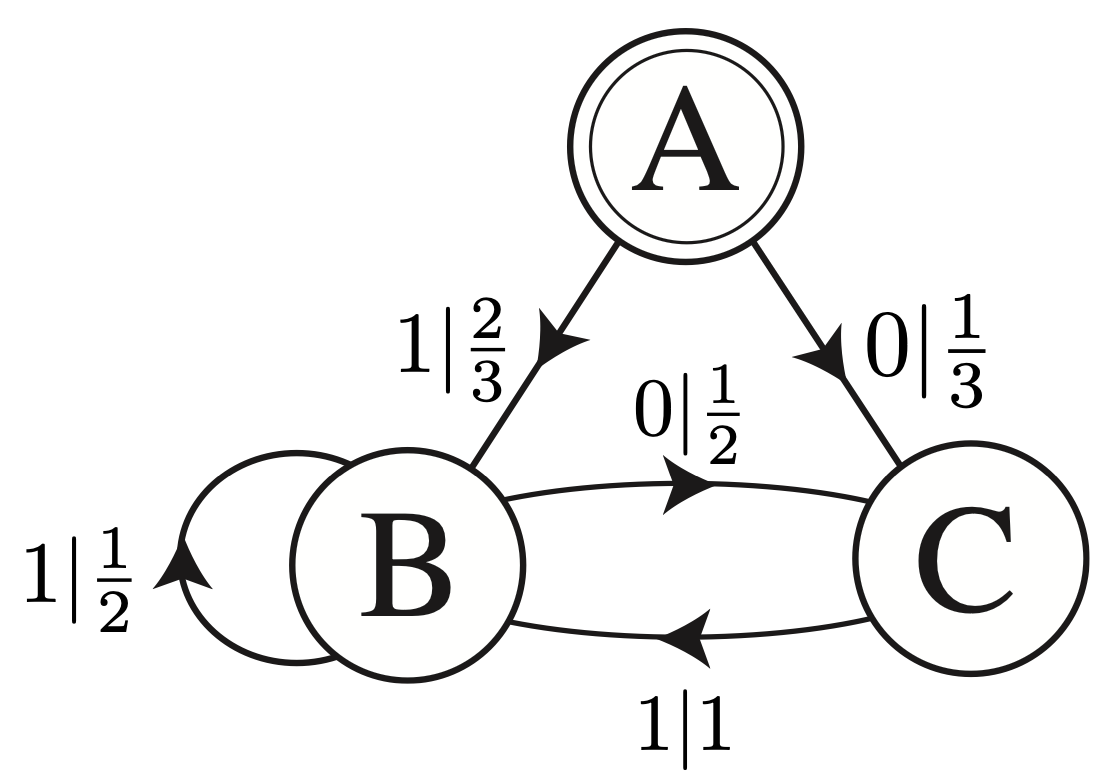} (d)
\includegraphics[width=.4\columnwidth]{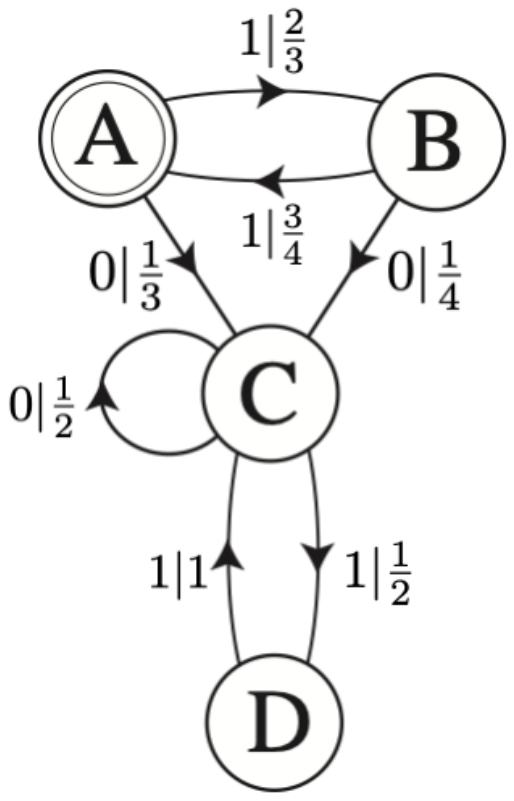} \caption{Example
\eMs for: (a) Biased Coin; (b) Period-$2$; (c) Golden Mean; and (d) Even
	Processes. Causal states are circles, with the start state having an
	inscribed circle. State-to-state transitions are labeled $x | p$, where
	$\meassymbol \in \alphabet$ is the emitted symbol and $p \in [0,1]$ is the
	transition probability.
	}
\label{fig:ExampleProcesses}
\end{figure*}

\begin{proposition}
If $\Process$ is stationary, then the prediction mapping $\epsilon_t
=\epsilon[\Past_t]$ is a random variable that is stationary.
\end{proposition}

\begin{proof}
$\epsilon_t$ is a first-order Markov process \cite{Shal98a}.
\end{proof}

The \emph{prediction process} $\MeasSymbol_\epsilon =
\{\MeasSymbol_t|\epsilon_t: t \in \mathbb{Z}\}$ is the stochastic process
of predictions $\Pr(\Future|\cdot)$.

\begin{proposition}
The prediction process is stationary and ergodic.
\end{proposition}

\begin{proof}
$\epsilon_t$ is a first-order Markov process \cite{Shal98a}.
\end{proof}

Similarly, the causal-state process is statistically well-behaved.

\begin{proposition}
If $\Process$ is stationary and ergodic, then the causal state process
$\{\CausalState_t: t \in \mathbb{Z}\}$ is stationary and ergodic.
\end{proposition}

\begin{proof}
The causal-state process is a finite-range function of the process. And so, by
previous propositions, since the process is stationary and ergodic, the
causal-state process is stationary and ergodic.
\end{proof}

Define the \emph{prediction uncertainty process}: $\hmu(t) =
\selficond{\msym_t}{\past_t}$. It is also given by: $\hmu(t) =
\selficond{\msym_t}{\CausalState_t}$. The latter expression is often a direct
and efficient way to generate the prediction process, if an \eM is in hand.

\begin{proposition}
The prediction uncertainty process is stationary and ergodic.
\end{proposition}

\begin{proof}
The prediction uncertainty process is a finite-range function of the process.
And so, by previous propositions, since it is stationary and ergodic.
\end{proof}

Define the \emph{causal-state uncertainty process}: $\Cmu(t) =
\selfii{\causalstate_t = \epsilon(\past_t)}$. This too is stationary and
ergodic if $\Process$ is:

\begin{proposition}
The causal-state uncertainty process is stationary and ergodic.
\end{proposition}

\begin{proof}
The causal-state uncertainty process is a finite-range function of the process.
And so, by previous propositions, since the latter is stationary and ergodic,
it is too.
\end{proof}

Having seen the pattern in the preceding results, we generalize. Let
$\SelfI(\cdot)$ be any function defining a temporal information atom---as those
laid out above in \Cref{tab:InfoAtomsProcess}. The preceding observations
generalize to establish that functions---in particular, the self-informations
$\SelfI(\Process)$---of stationary, ergodic processes are stationary and
ergodic.

\begin{proposition}
Self-information processes $\SelfIcond{A}{\overline{A}}(t)$ are stationary
and ergodic, if the process is.
\end{proposition}

\begin{proof}
The self-information processes are finite-range functions of the process. And
so, by previous propositions, they stationary and ergodic.
\end{proof}

This is all to say that an agent has well-defined and statistically well-behaved
informational quantities to work with in its more sophisticated downstream
cognitive processing, such as further-derived information-theoretic quantities,
in making inferences about them, in quantitatively monitoring its inferences,
in identifying fluctuations, and in decision-making and taking actions on the
environment.


\begin{table*}[!htbp]
\renewcommand{\arraystretch}{1.1}
\begin{tabular}{|c|c|c|c|c|}
\hline
\multicolumn{5}{|c|}{Observer's Analysis of Biased Coin Process} \\
\hline
State & Measurement $\msym$ & Surprise $-\log_2 \Pr(\msym)$ [bits] & Semantic
State $\causalstate$: Meaning & Degree of Meaning $\Theta(\msym)$ [bits] \\
\hline
$A$ & $\lambda$ & Not Defined & No Measurement & $0.0$ \\
$A$ & $1$ & $0.585$ & $A$: Sync & $0.0$ \\
$A$ & $0$ & $1.585$ & $A$: Sync & $0.0$ \\
\hline
\end{tabular}
\caption{Information and Semantic Analysis of the Biased Coin Process, with
	$1$s-bias of $\Pr(\msym = 1) = 2/3$. Recall Fig.
	\ref{fig:ExampleProcesses}(a). }
\label{tab:InfoSemanticAnalysisBiasedCoin}
\end{table*}


\begin{figure*}
\includegraphics[width=.987\textwidth]{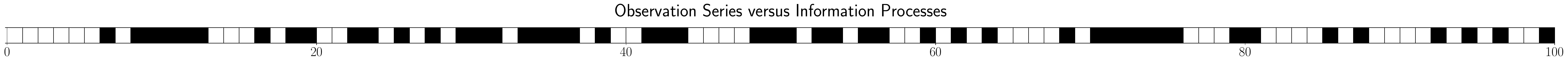}
\includegraphics[width=\textwidth]{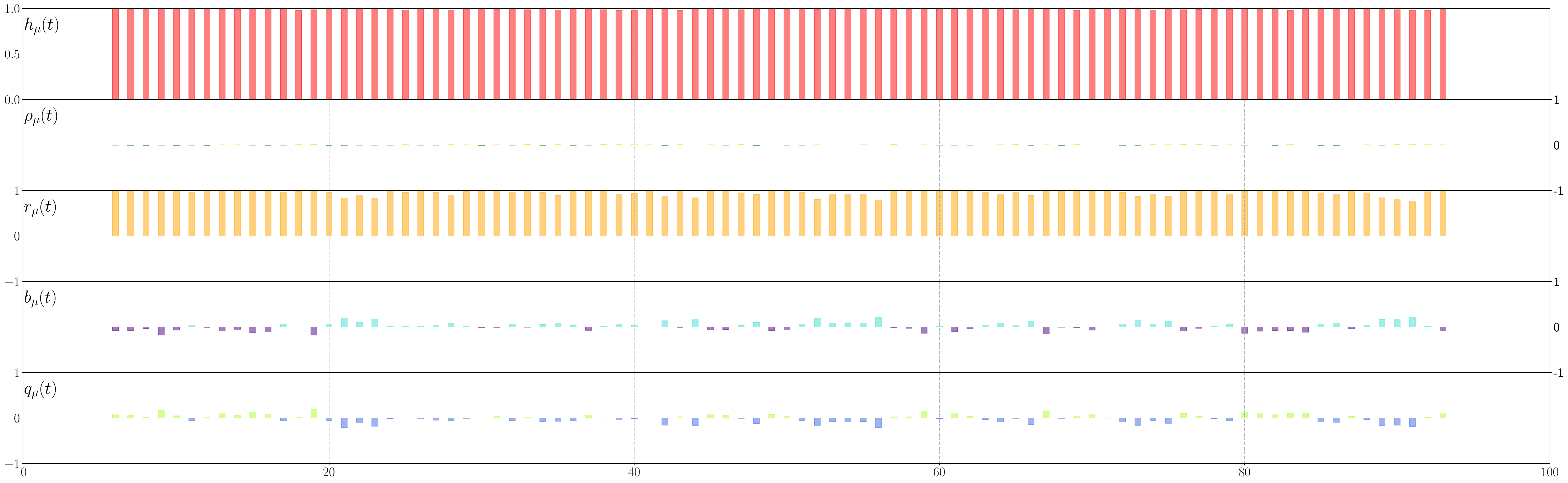}
\caption{Biased Coin Process observed time series realization (top) versus 
	several of its information processes (below): entropy rate $\hmu(t)$,
	anticipated information rate $\rhomu(t)$, ephemeral rate $\rmu(t)$, bound
	information rate $\bmu(t)$, and enigmatic information rate $\qmu(t)$.
	The information process values for the first and last $6$ steps are not
	shown, due to their being estimated using a window length of $13$: $6$
	steps for the history and for the future words, plus the present ($1$
	step). This holds for all of the following information process time series.
	}
\label{fig:iProcess_BC}
\end{figure*}

\section{Example Information Processes}
\label{sec:IProcessExamples}

The following analyzes concrete examples of information processes produced an
agent observing an environment governed by finite-memory generators: (i) wholly
unpredictable, (ii) predictable, (ii) structured, but finite Markov order, and
(iv) structured and infinite Markov order. In other words, several relatively
low-complexity (easy to explain) environments that range from predictable
memoryless to unpredictable infinite-range dependencies---a suite that
illustrates the breadth of basic results of interest. See Fig.
\ref{fig:ExampleProcesses} for their state transition diagrams.


Building on these and the general development of information processes, the
following section establishes that, whatever its role, an agent has access to
and generates well-behaved information processes---first, such derived
processes are stationary and second they are ergodic. Tables
\ref{tab:InfoSemanticAnalysisBiasedCoin}, \ref{tab:InfoSemanticAnalysisPeriod2},
\ref{tab:InfoSemanticAnalysisGoldenMean}, and
\ref{tab:InfoSemanticAnalysisEven}
present an agent's informational and semantic analysis of the example processes,
respectively. And, Figs. \ref{fig:iProcess_BC}, \ref{fig:iProcess_P2},
\ref{fig:iProcess_GMP}, and \ref{fig:iProcess_EP} plot typical realizations of
the examples' information processes.

References \cite{Crut24a,Jurg25a} provide fuller informational analyses for
these and other example processes. Here, we simply focus on aspects related to
a cognitive agent's associated information processes.

\begin{table*}[!htbp]
\renewcommand{\arraystretch}{1.1}
\setlength\tabcolsep{4pt}
\begin{tabular}{|c|c|c|c|c|}
\hline
\multicolumn{5}{|c|}{Observer's Analysis of Period-$2$ Process} \\
\hline
State & Measurement $\msym$ & Surprise $-\log_2 \Pr(\msym)$ [bits] & Semantic
State $\causalstate$: Meaning & Degree of Meaning $\Theta(\msym)$ [bits] \\
\hline
$A$ & $\lambda$ & Not Defined & No Measurement & $0.0$ \\
$A$ & $1$ & $1.0$ & $B$: Sync & $1.0$\\
$A$ & $0$ & $1.0$ & $C$: Sync & $1.0$ \\
\hline
$B$ & $1$ & $\infty$ & $A$: Loose sync; reset & $0.0$ \\
$B$ & $0$ & $0.0$ & $C$: Deterministic $1$ & $1.0$ \\
\hline
$C$ & $1$ & $0.0$ & $B$: Deterministic $0$ & $1.0$ \\
$C$ & $0$ & $\infty$ & $A$: Loose sync; reset & $0.0$ \\
\hline
\end{tabular}
\caption{Information and Semantic Analysis of the Period-$2$ Process. Recall
	Fig. \ref{fig:ExampleProcesses}(b). }
\label{tab:InfoSemanticAnalysisPeriod2}
\end{table*}

\begin{figure*}
\includegraphics[width=.987\textwidth]{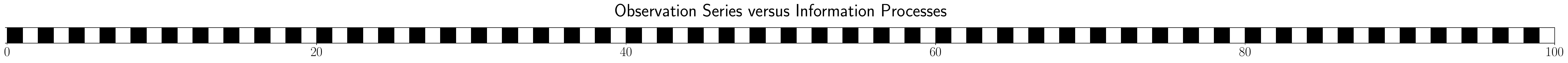}
\includegraphics[width=\textwidth]{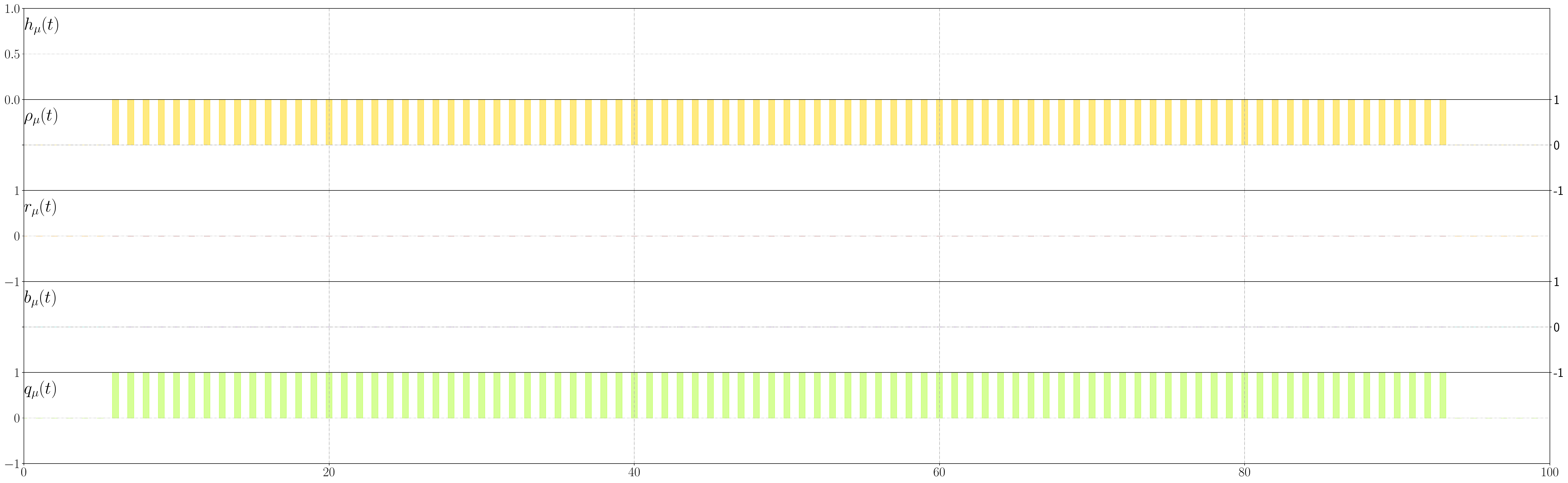}
\caption{Period-$2$ Process observed time series sample versus several of its
	information processes: entropy rate $\hmu(t)$, anticipated information
	rate $\rhomu(t)$, ephemeral rate $\rmu(t)$, bound information rate
	$\bmu(t)$, and enigmatic information rate $\qmu(t)$.
	}
\label{fig:iProcess_P2}
\end{figure*}

\subsection{Unpredictable: Independent Identically-Distributed}
\label{sec:BiasedCoin}
\label{sec:IID}

The paradigm of a random process is the fair (or biased) coin---one in the
family of \emph{independent identically distributed} (IID) processes---coins,
dies, and the like. Consider the particular case of a process of identical
independently distributed variables; i.e., a coin flip with probability
$\left\{ \Pr(H) = p, \Pr(T) = 1-p \right\}$ repeated infinitely many times.

Figure \ref{fig:ExampleProcesses}(a) gives the \eM state-transition diagram.
For all processes in the IID family, there is only a single causal state
$\CausalStateSet = \{A\}$. For the binary alphabet process generated the
state-to-state transitions are given by two $1\times 1$ symbol-labeled
transition matrices: $\{ T^{(0)} = (p), T^{(1)} = (1-p) \}, p \in [0,1]$.  The
initial measure over states is also trivial: $\pi_0 = \{1\}$. $\pi_t$ is
invariant over time.

Its basic informational quantities are: (i) entropy rate $\hmu = 1$ bit per time
step, (ii) statistical complexity $\Cmu = 0$ bits, and (iii) excess entropy $\EE
= 0$ bits. All as one would expect for a fully random, memoryless process.

Each measurement $\MeasSymbol_t$ being independent, the mutual information
$\SelfImut{\MeasSymbol_t}{\MeasSymbol_{t^{\prime}}} = 0$ for all $t \neq
t^{\prime}$. Therefore, all information atoms in
\cref{fig:iDiagramProcess}(left) vanish except for $\rmu = \hmu =
\SelfII{\MeasSymbol_t}$ and the infinite amount of information in the past and
future. 


Table \ref{tab:InfoSemanticAnalysisBiasedCoin} presents an agent's
informational and semantic analysis of an IID process in terms of the
self-information (transition) surprise $-\log_2(\msym)$ and degree of meaning
$\Theta(\msym)$: IID processes are always synchronized and the semantics is
meaningless, as one would expect for a structureless process.

Finally, there are the Biased Coin's information processes. Figure
\ref{fig:iProcess_BC} presents a realization and its information processes:
entropy rate $\hmu(t)$, anticipated information rate $\rhomu$, ephemeral rate
$\rmu$, bound information rate $\bmu(t)$, and enigmatic information rate
$\qmu(t)$. As expected, every succeeding observation is highly surprising
($\hmu = 1$ bit) and no information in the future is anticipated ($\rhomu = 0$
bits). All of the information $\hmu$ produced at each time step is forgotten
($\rmu \approx 1$) and none is stored ($\bmu \approx 0$). The fluctuations seen
in these latter information processes reflect their being empirically estimated
from a finite time series of observations (length $\ell = 10^6$ symbols.)

\begin{figure*}
\includegraphics[width=.987\textwidth]{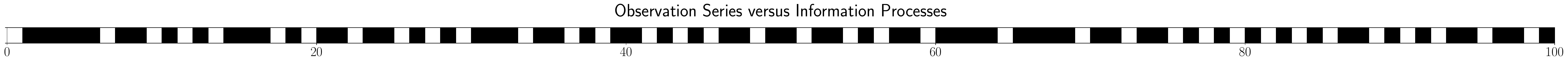}
\includegraphics[width=\textwidth]{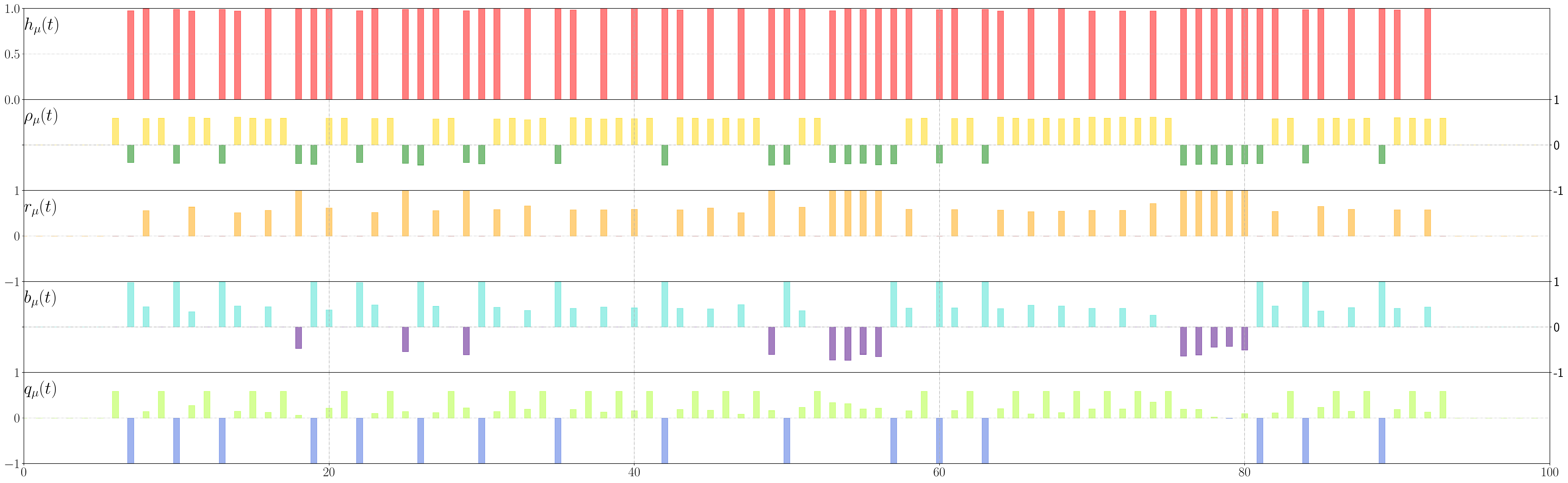}
\caption{Golden Mean Process observed time series sample versus several of its
	information processes: entropy rate $\hmu(t)$, anticipated information
	rate $\rhomu(t)$, ephemeral rate $\rmu(t)$, bound information rate
	$\bmu(t)$, and enigmatic information rate $\qmu(t)$.
	}
\label{fig:iProcess_GMP}
\end{figure*}

\begin{table*}[!htbp]
\renewcommand{\arraystretch}{1.1}
\begin{tabular}{|c|c|c|c|c|}
\hline
\multicolumn{5}{|c|}{Observer's Analysis of Golden Mean Process} \\
\hline
State & Measurement $\msym$ & Surprise $-\log_2 \Pr(\msym)$ [bits] & Semantic
State $\causalstate$: Meaning & Degree of Meaning $\Theta(\msym)$ [bits] \\
\hline
$A$ & $\lambda$ & Not Defined & No Measurement & $0.0$ \\
$A$ & $1$ & $0.585$ & $B$: Sync & $0.585$ \\
$A$ & $0$ & $1.585$ & $C$: Sync & $1.585$ \\
\hline
$B$ & $1$ & $1.0$ & $B$: Random $1$ & $0.585$ \\
$B$ & $0$ & $1.0$ & $C$: Isolated $0$ & $1.585$ \\
\hline
$C$ & $1$ & $0$ & $B$: Deterministic $1$ & $0.585$ \\
$C$ & $0$ & $\infty$ & $A$: Loose sync; Reset & $0.0$ \\
\hline
\end{tabular}
\caption{Information and Semantic Analysis of the Golden Mean Process. Recall
	Fig. \ref{fig:ExampleProcesses}(c). }
\label{tab:InfoSemanticAnalysisGoldenMean}
\end{table*}

\subsection{Predictable: Period-\texorpdfstring{$n$}{n}}
\label{sec:Periodn}

Now consider the other extreme of predictability---a completely determined
period-$2$ process $\dots 0 1 0 1 0 1 \dots$. Figure
\ref{fig:ExampleProcesses}(b) presents the \eM that generates a period-$2$
binary process. There is a single transient state $A$ and the initial
distribution is $\pi = \{1,0,0\}$. The recurrent states $B$ and $C$ are equally
likely---$\pi = \{0,1/2,1/2\}$---which state distribution is time invariant.

The state transition matrices for the binary generator are:
\begin{align*}
T^{(0)} =
\begin{pmatrix}
0 & 1/2 & 1/2 \\
0 & 0   & 1   \\
0 & 0   & 0 
\end{pmatrix} ~\text{and}~
T^{(1)} =
\begin{pmatrix}
0 & 1/2 & 1/2 \\
0 & 0   & 0   \\
0 & 1   & 0 
\end{pmatrix}
  ~.
\end{align*}

All conditional entropies vanish, including the information in the past and
future, as any single measurement determines the value of $\MeasSymbol_t$ for
all time indices. The only remaining quantity is $\qmu$, which contains all the
process information. In the case of a period-$n$ process, the value of this
atom in bits is exactly $\log_2 n$.
Table \ref{tab:InfoSemanticAnalysisPeriod2} gives an agent's informational and
semantic analysis of a period-$2$ process. When disallowed symbols violate the
periodicity, the \eM resets to the (meaningless) start state, having lost the
oscillation's phase, while the agent is infinitely surprised that this has
occurred.


Finally, there are Period-$2$'s information processes. Figure
\ref{fig:iProcess_P2} presents a realization and its information
processes: entropy rate $\hmu(t)$, anticipated information rate
$\rhomu$, ephemeral rate $\rmu$, bound information rate $\bmu(t)$, and
enigmatic information rate $\qmu(t)$. As expected, every succeeding
observation is determined ($\hmu = 0$ bit) and all of this future
information is anticipated ($\rhomu = 1$ bits). Since no information
$\hmu$ produced at each time step there is none to forget ($\rmu
\approx 0$) and none to store stored ($\bmu \approx 0$).  There is a
single bit ($\qmu = 1$ bit) of phase information of the period-$2$
cycle.

The extreme examples of utterly random and perfectly predictable are simple
enough to work through, but typically calculating the information dynamics
requires taking entropies over infinite sequences, which is not practical.
Therefore, we turn to the use of optimal finite models and the theory of
computational mechanics. 

\subsection{Finite Markov Order: Golden Mean}
\label{sec:GMeMachine}

As one sees from Ref. \cite{Jame10a}'s survey (Fig. 13 there), the vast majority
of structured stochastic processes lie between the two preceding extremes of
predictability. As a typical example of these intermediate-complexity process
generators consider the Golden Mean Process. See Fig.
\ref{fig:ExampleProcesses}(c) for its \eM state transition diagram, which
generates all binary sequences except for those with consecutive $0$s. That is,
only the word $w = 00$ is forbidden; otherwise the generated realizations are
random. It is easy to see where in the state-transition diagram this restriction
arises: When the \eM is in state $C$ it must transition to state $B$ and emit a
$1$ with probability $1$. A simple way to summarize this and so characterize the
Golden Mean Process is to give its list of \emph{irreducible forbidden words}
$\mathcal{F} = \{00\}$.

As in the period-$2$ process, there is a single transient state $A$ and the
initial distribution is $\pi_0 = \Pr(A,B,C) = \{1,0,0\}$. The state transition
matrices for the Golden Mean Process generator are:
\begin{align*}
T^{(0)} =
\begin{pmatrix}
0 & 0 & 1/3 \\
0 & 0 & 1/2   \\
0 & 0 &  0 
\end{pmatrix}
  ~\text{and}~
T^{(1)} =
\begin{pmatrix}
0 & 2/3 & 0 \\
0 & 1/2 & 0 \\
0 & 1   & 0 
\end{pmatrix}
  ~.
\end{align*}
The recurrent states $B$ and $C$ are visited with probabilities $\widehat{\pi} =
\{0,2/3,1/3\}$---which state distribution is time invariant after the first time
step.

The informational analysis says the entropy rate is $\hmu = 2/3 ~\SelfI(1/2) =
2/3$ bits per emission. The statistical complexity is $\Cmu = \SelfI(2/3)$
bits, as is the excess entropy. The generated Golden Mean process is Markov
order $1$. For these see Ref. \cite{Jame11a}.

Table \ref{tab:InfoSemanticAnalysisGoldenMean} gives an agent's informational
and semantic analysis of the Golden Mean Process.

And, Fig. \ref{fig:iProcess_GMP} presents its information processes: entropy
rate $\hmu(t)$, anticipated information rate $\rhomu$, ephemeral rate $\rmu$,
bound information rate $\bmu(t)$, and enigmatic information rate $\qmu(t)$.


As concrete examples of the Golden Mean Process' information processes we
examine the prediction and statistical complexity processes---$\hmu(t)$ and
$\Cmu(t)$. First, we need the state distribution $p(t)$ as a function of time:
this is simply $p(0) = (1,0,0)$ and $p(t) = (0,2/3,1/3), t = 1, 2, 3, \ldots)$.

From this, we then calculate the time-local quantities from looking at machine
states---out of which set of values this or that information process will
consist of a time series of samples. That is, for the prediction process we
have uncertainties in state $A$ for first time step: for observing a $1$
$\hmu(t=0,\msym_0 = 1) = \SelfI(2/3)$, a $0$ $\hmu(t=0,\msym_0 = 0) =
\SelfI(1/3)$. For second step; $\hmu(t=1,\msym_{0:1}  = 11 ) = 1$ bit,
$\hmu(t=1,\msym_{0:1}  = 10 ) = 0$ bit. $\hmu(t=1,\msym_{0:1}  = 01 ) = 0$ bit.
At this point, having these initial observations, the agent is synchronized to
the environment behavior. And so, values in the prediction process are samples
of these values, sampled according to the history seen. Then we have the
statistical complexity process: $\Cmu(t=0) = 0$. $\Cmu(t) = \SelfI(2/3)$, $t =
1, 2, 3, \ldots$. This is quite simple.

\begin{table*}[!htbp]
\renewcommand{\arraystretch}{1.1}
\begin{tabular}{|c|c|c|c|c|}
\hline
\multicolumn{5}{|c|}{Observer's Analysis of Even Process} \\
\hline
State & Measurement $\msym$ & Surprise $-\log_2 \Pr(\msym)$ [bits] & Semantic
State $\causalstate$: Meaning & Degree of Meaning $\Theta(\msym)$ [bits] \\
\hline
$A$ & $\lambda$ & Not Defined & No Measurement & $0.0$ \\
$A$ & $1$ & $0.585$ & $B$: Unsync & $0.585 \ldots \infty$ \\
$A$ & $0$ & $1.585$ & $C$: Sync & $1.585 \ldots 0.585$ \\
\hline
$B$ & $1$ & $0.415$ & $A$: Unsync & $0.585 \ldots \infty$ \\
$B$ & $0$ & $1.585$ & $C$: Sync & $1.585 \ldots 0.585$ \\
\hline
$C$ & $1$ & $1$ & $D$: Odd \#$1$s & $1.585$ \\
$C$ & $0$ & $1$ & $C$: Even \#$1$s & $0.585$ \\
\hline
$D$ & $1$ & $0$ & $C$: Even \#$1$s & $0.585$ \\
$D$ & $0$ & $\infty$ & $A$: Loose sync; Reset & $0.0$ \\
\hline
\end{tabular}
\caption{Information and Semantic Analysis of the Even Process. Recall Fig.
	\ref{fig:ExampleProcesses}(d). (Reproduced with permission from Ref.
	\cite{Crut91b}.)}
\label{tab:InfoSemanticAnalysisEven}
\end{table*}

\begin{figure*}
\includegraphics[width=.987\textwidth]{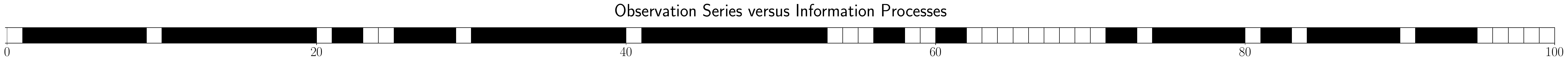}
\includegraphics[width=\textwidth]{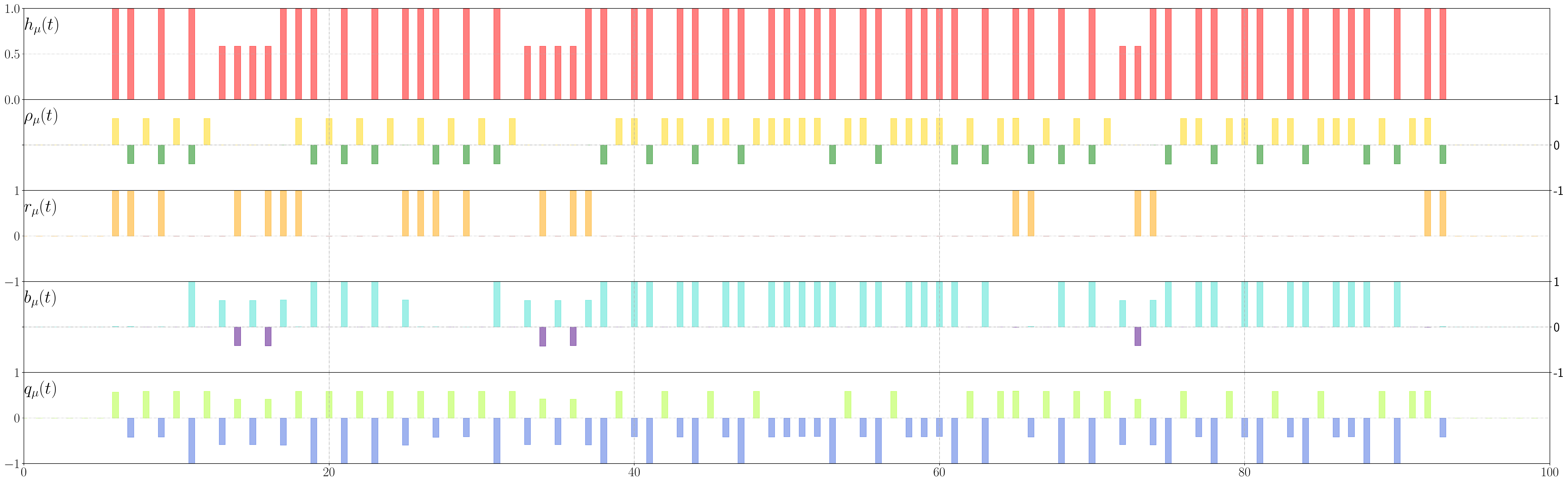}
\caption{Even Process observed time series sample versus several of its
	information processes: entropy rate $\hmu(t)$, anticipated information
	rate $\rhomu(t)$, ephemeral rate $\rmu(t)$, bound information rate
	$\bmu(t)$, and enigmatic information rate $\qmu(t)$.
	}
\label{fig:iProcess_EP}
\end{figure*}

Finally, there are Golden Mean's information processes. Figure
\ref{fig:iProcess_GMP} presents a realization and its information
processes: entropy rate $\hmu(t)$, anticipated information rate
$\rhomu$, ephemeral rate $\rmu$, bound information rate $\bmu(t)$, and
enigmatic information rate $\qmu(t)$.

Given that there are no consecutive $0$s, if a $0$ is observed there is no
uncertainty in the next observation (it must be a $1$). At those times (denote
them $t$), $\hmu(t) = 0$. Otherwise, (denote those times $t^\prime$), a fair
coin flip is anticipated: $\hmu(t^\prime) = 1$ bit. At times $t^\prime$ the bit
of created information is either completely lost ($\rmu(t^\prime) = 1$ bit) or
some of it is stored ($\bmu(t^\prime) > 0$) such that $\rmu + \bmu = \hmu$.

We notice there are several negative informations---$\bmu(t)$, $\rhomu$, and
$\qmu$. As explained at the end of Sec. \ref{Sec:DynamicIProcesses}, negative
informations are to be expected in the setting of temporal information measures.

Overall, we see how the various information processes make clear that the
Golden Mean Process is a complicated combination of the two information
processing modes of the Biased Coin and Periodic Processes. There is more to
say, of course. What is missing is a more mechanistic explanation of how the
information quantities are being stored and processed. This is the burden of a
sequel. Here, our goal was to motivate and then explore information processes.
The next example illustrates the need for this amply.

\subsection{Infinite Markov Order: Even}
\label{sec:EPProcess}

The fourth example to consider is more complex still. This is the Even Process,
whose \eM is shown in Fig. \ref{fig:ExampleProcesses}(d). Whereas the Golden
Mean Process was determined by finite-length word restrictions, the Even
Process exhibits infinite-range statistical dependencies, despite being finite
state. The Even Process consists of words in which $1$s occur in blocks (pairs)
of even length bounded by $0$s. The evenness criterion is a statistical
dependency of infinite range. One consequence is that the Even Process has
infinite Markov order. Moreover, its irreducible forbidden set is countably
infinite: $\mathcal{F} = \{ 01^{2n+1}0 : n \in \mathbb{Z}_{\geq 0} \}$.

There are two transient states $A$ and $B$ and the initial distribution is
$\pi_0 = \{1,0,0,0\}$. The state transition matrices for the generator are:
\begin{align*}
T^{(0)} =
\begin{pmatrix}
0 & 0 & 1/3 & 0 \\
0 & 0 & 1/4 & 0 \\
0 & 0 & 1/2 & 0 \\
0 & 0 &  0  & 0 \\
\end{pmatrix}
  ~\text{and}~
T^{(1)} =
\begin{pmatrix}
 0  & 2/3 &  0  &  0 \\
3/4 &  0  & 1/4 &  0 \\
 0  &  0  &  0  & 1/2 \\
 0  &  0  &  0  &  0 \\
\end{pmatrix}
  .
\end{align*}
The recurrent states $C$ and $D$ are visited with probabilities $\widehat{\pi} =
\{0,0,2/3,1/3\}$---which state distribution is time invariant after an infinite
number of observations.

The informational analysis says the entropy rate is $\hmu = 2/3 ~\SelfI(1/2) =
2/3$ bits per emission. The statistical complexity is $\Cmu = \SelfI(2/3)$
bits, as is the excess entropy. These are all the same as for the Golden Mean
Process. Nonetheless, it is clear that the Even Process, being infinite-Markov
is qualitatively different from the Golden Mean Process. For related analyses
see Ref. \cite{Jame11a}.

To start analyzing the Even Process' information processes, as above, we
determine the time-dependent state distribution $p(t)$ explicitly:
\begin{align*}
p(t = 0)   & = \big( 1,0,0,0 \big) , \\
p(t = 1)   & = 
\begin{matrix}
\big( \Pr(\CausalState|\ms{0}{0} = 0)  & = (0,1/3,0,0) ~, \\
  \Pr(\CausalState|\ms{0}{0} = 1)  & = (0,0,2/3,0) \big)
\end{matrix} \\
p(t = 2) & = 
\begin{matrix}
\big( \Pr(\CausalState|\ms{0}{1} = 00) & = (0,0,1/6,0) ~, \\
\Pr(\CausalState|\ms{0}{1} = 01) & = (0,0,0,1/6) ~, \\
\Pr(\CausalState|\ms{0}{1} = 10) & = (0,0,1/6,0) ~, \\
\Pr(\CausalState|\ms{0}{1} = 11) & = (1/2,0,0,0) \big)
\end{matrix} \\
\\
\ldots
  ~.
\end{align*}
This illustrates the transient phase---how the initial state probabilities begin
to settle toward the asymptotic distribution $\widehat{\pi}$. This relaxation
takes infinite time; reflecting the processes infinite Markov order. And this,
in turn, means that the information processes---such as, $\hmu(t)$ and $\Cmu(t)$
also take infinite-time to reach their asymptotic values.

Table \ref{tab:InfoSemanticAnalysisEven} presents an agent's informational and
semantic analysis of the Even Process. It shows, in particular, that the degree
of meaning continues to change during this relaxation. In fact, the information
processes do not reach their asymptotic behaviors until the agent is
synchronized. And, this only occurs once a $0$ is observed---taking the state to
$C$.



As above, we leave explicitly calculating the companion information processes
as an exercise. Nonetheless, Fig. \ref{fig:iProcess_EP} presents several of its
information processes: entropy rate $\hmu(t)$, anticipated information rate
$\rhomu$, ephemeral rate $\rmu$, bound information rate $\bmu(t)$, and
enigmatic information rate $\qmu(t)$.


In one sense, the Even Process is markedly more complex than the Golden Mean
Process as it has infinite-range statistical properties. Specifically, while
the Golden Mean as a single forbidden word $\mathcal{F} = {00}$, the Even
Process has a countable infinity $\mathcal{F} = \{0(11)^n0: n = 1, 2,
\ldots\}$. In practical terms this enhances the complication of the
informational narrative of the information processes given in Fig.
\ref{fig:iProcess_EP}.



\begin{table*}[!htbp]
\renewcommand{\arraystretch}{1.1}
\begin{tabular}{|c|c|c|c|c|}
\hline
\multicolumn{5}{|c|}{Semantics of the Period-$2$ Process $X$ as the Golden Mean
Process $X^\prime$} \\
\hline
State & Measurement $\msym^\prime$ & $\Delta$ Surprise $\log_2 \left(
\Pr(\msym)/\Pr(\msym^\prime) \right)$ & State $\causalstate^\prime$: Meaning &
$\Delta \Theta = \log_2 \left(
\Pr(\causalstate)/\Pr(\causalstate^\prime)\right)$ \\
\hline
$A'$ & $\lambda$ & 0.0 & No Measurement & $0$ \\
$A'$ & $1$ & $-0.415 = \log_2 \left( (1/2)/(2/3) \right)$ & $B'$: Sync & $-0.415
= \log_2 \left( (1/2)/(2/3) \right)$  \\
$A'$ & $0$ & $0.585 = \log_2 \left( (1/2)/(1/3) \right)$ & $C'$: Sync & $0.585 =
\log_2 \left( (1/2)/(1/3) \right)$ \\
\hline
$B'$ & $1$ & $ -\infty = \log_2 \left( (0)/(1/2) \right)$ &  & Never Occurs \\
$B'$ & $0$ & $ 1 = \log_2 \left( (1)/(1/2) \right)$ & $C'$: Isolated $0$ & $1$
\\
\hline
$C'$ & $1$ & $ 0 = \log_2 \left( (1)/(1) \right)$  & $B'$: Deterministic $1$ &
$0.585$ \\
$C'$ & $0$ & $ -\infty = \log_2 \left( (0)/(1) \right)$ & & Never Occurs \\
\hline
\end{tabular}
\caption{Semantics of Period-$2$ Process as if the Golden Mean Process. Recall
	Figs. \ref{fig:ExampleProcesses}(b) and (c). }
\label{tab:P2asGMPIncorrectSemantics}
\end{table*}

\begin{table*}[!htbp]
\renewcommand{\arraystretch}{1.1}
\begin{tabular}{|c|c|c|c|c|}
\hline
\multicolumn{5}{|c|}{Semantics of Even Process $X$ as the Golden Mean Process
$X^\prime$ } \\
\hline
State & Measurement $\msym^\prime$ & $\Delta$ Surprise $\log_2 \left(
\Pr(\msym)/\Pr(\msym^\prime) \right)$ & State $\causalstate^\prime$: Meaning &
$\Delta \Theta = \log_2 \left(
\Pr(\causalstate)/\Pr(\causalstate^\prime)\right)$ \\
\hline
$A'$ & $\lambda$ & Not Defined & No Measurement & $0$ \\
$A'$ & $1$ & $0$ & $B'$: Sync & $0$ \\
$A'$ & $0$ & $0$ & $C'$: Sync & $0$ \\
\hline
$B'$ & $1$ & $ 0.585 = \log_2 \left( (3/4)/(1/2) \right)$ & $B'$: Random $1$ & $
-0.415 = \log_2 \left( (1/2)/(2/3) \right)$ \\
$B'$ & $0$ & $ -1 = \log_2 \left( (1/4)/(1/2) \right)$ & $C'$: Isolated $0$ & $
0 = \log_2 \left( (1/3)/(1/3) \right)$ \\
\hline
$C'$ & $1$ & $-1 = \log_2 \left( (1/2)/(1) \right)$ & $B'$: Deterministic $1$ &
$0.585$ \\
$C'$ & $0$ & $\infty = \log_2 \left( (1/2)/(0) \right)$ & $A'$: Loose sync;
Reset & $0$ \\
\hline
\end{tabular}
\caption{Semantics of Even Process as if the Golden Mean Process. Recall Figs.
	\ref{fig:ExampleProcesses}(c) and (d). }
\label{tab:EPasGMPIncorrectSemantics}
\end{table*}

In short, the Golden Mean's order-$1$ Markov order translates into the
ability to track short words to determine the level of next-symbol
unpredictability. Equivalently, in this case, tracking these short words
allows one to know in which of the Golden Mean's causal states ($B$ and
$C$) the process is in at any given time.

This is not so for the Even Process. There are statistical dependencies of
arbitrary length involved in knowing which state the process is. A quantitative
consequence is the information measures at any given time $t$ are not
simply-interpreted values. However, we can easily extract interpretations of
the Even Process' information process values---sometimes exactly, but typically
only approximately. (That is, short of writing down these measures' closed form
expressions using the methods of Ref. \cite{Crut13a}, which is the mandate for a
sequel.)

Briefly, the easy interpretations are expedited by referring to the Even
Process' state transition diagram in Fig. \ref{fig:ExampleProcesses}(d). First,
we consider only asymptotic statistics by ignoring transient states $A$ and $B$
and concentrating on recurrent states $C$ and $D$. Practically, one simply uses
a realization far after the start state.

Notice that  when in state $C$ the next observation is a fair coin flip:
$\hmu(t) = 1$ bit. And, when in state $D$ the next observation is wholly
determined (it occurs certainly) to be a $1$. It is the second $1$ in the Even
Process' characteristic pairs of $1$s. That is, $\hmu(t) = 0$ bit. Due to this
$\rmu(t)$ and $\bmu(t)$ vanish, as seen in the information process plot.
One can go much further, of course, but that is the burden of a sequel.


As with the Golden Mean Process, we again see the expected negative information
measures. And, that the Even Process is a markedly more complex combination of
the two information processing modes---Biased Coin and Periodic
Processes-illustrated by the Golden Mean Process. Again, a sequel will provide
a directly mechanistic explanation of how the information measures are being
stored and transformed, making it clearer how the Even Process is more complex.
Again, the proceeding's goal was an introduction to and initial exploration of
information processes.

\subsection{Misdirected Semantics}
\label{sec:GMInterpEP}

These example analyses all assumed the agent knew what the process generator
was and that the agent came to be synchronized to the environment state.
However, the semantic theory here applies more generally, to circumstances when
the agent has an approximate or even incorrect internal model of its
environment.

To be concrete, consider two examples of misinterpretation: An agent uses the
Golden Mean \eM to interpret the semantics of (i) the Period-$2$ Process and
(ii) the Even Process. (Recall Figs. \ref{fig:ExampleProcesses}(b), (c), and
(d).) The analyses are presented in Tables \ref{tab:P2asGMPIncorrectSemantics}
and \ref{tab:EPasGMPIncorrectSemantics}.

We simply track how both processes, starting in their start state $A$, produce
words up to length $2$: $\omega \in \{00, 0,1 10, 11\}$. The Biased Coin
recognizes all; the Golden Mean disallows $00$.

We use the information gain (or Kullback-Leibler divergence) to give the
self-information difference between the expected next symbol versus actual and
information distance between degree of meanings of interpreted state and actual
state. The first relates to differing expectations in predicting symbols and
transition probabilities---the \emph{surprise} $\Delta$:
\begin{align*}
\Delta = \log_2 \frac{\Pr(\msym)}{\Pr(\msym^\prime)}
  ~.
\end{align*}
The second concerns the differing interpretations of the state
semantics---$\Delta \Theta$:
\begin{align*}
\Delta \Theta = \log_2 \frac{\Pr(\to_\msym
\causalstate)}{\Pr(\to_{\msym^\prime} \causalstate^\prime)}
  ~.
\end{align*}
Note that when comparing a process to itself these vanish: $\Delta= 0$ and
$\Delta \Theta = 0$.

\section{Conclusion}
\label{sec:conclusion}

Many complex adaptive systems consist of an individual agent or a collection
of interacting agents that take in and process data from their surroundings and
then take actions based on what is gleaned. The preceding addressed only half
of this, What kinds of stochastic process are involved as an agent monitors
its environment? Our goal was to call out what is common---from an
information-theoretic perspective---across such systems.

In this, we introduced information processes---time series of various kinds of
Shannon information measure that capture a range of signal types from degrees
of unpredictability to degrees of correlation and temporal memory. As a
technical focus, we explored conditions for these signals' stationarity and
ergodicity. These are properties that strongly affect what an agent can use
statistically ``downstream'' to make reliable estimates of environmental
conditions and make robust decisions and take actions necessary for
functioning.

Simply having well-behaved informational signals in hand, though, is far from
the whole story of agent functioning. The question immediately arises as to
what the signals mean to an agent and what the agent can do with them. We
addressed the former by reviewing Shannon information measures, following on
Refs. \cite{Jame11a,Jame13a}'s analysis of structured stochastic processes.
This suite of measures has natural, informational, even functional
meanings---they give an intrinsic semantics of an environment's behavior, as an
agent experiences it. These statistics---environment unpredictability, the
present's correlation with the past, memory required for optimal prediction,
and the amount of hidden internal-state information---embody key kinds of
information that an agent needs in order to properly operate.

Taken together then, information processes, their ergodicity, and their
semantics lay a foundation for agentic information theory. That is, they
comprise a first step to move beyond asymptotic properties to real-time signals
needed for functioning, flourishing complex adaptive systems.

\section*{Acknowledgments}

The authors thank the Telluride Science Research Center for hospitality during
visits and the participants of the Information Engines Workshops there. We
thank Ryan James for his adept use of his Python \emph{dit} discrete
information theory package \cite{dit}, Greg Wimsatt for comments on measure
theory and stochastic processes, and Alex Boyd for the real-time thermodynamic
data as a physical example of an information process. We thank them all for
helpful discussions. This material is based upon work supported by, or in part
by, the Art and Science Laboratory and the U.S. Army Research Laboratory and
the U.S. Army Research Office under grant W911NF-21-1-0048.

\appendix 

\newtheorem*{proposition*}{Proposition}

\section{Function of a Process is a Process}
\label{app:FofP}

\begin{proposition*}
Let $\{X(t)\}_{t \in \mathbb{Z}}$ be a stochastic process on probability space
$(\Omega,F,P)$ and $f:\mathbb{R} \to \mathbb{R}$ be a measurable function.
Then $Y = \{ Y(t) = f(X(t)), t \in \mathbb{Z}\}$, is a stochastic process.
\end{proposition*}

\begin{proof}
First, $Y(t)$ is a random variable for each $t \in \mathbb{Z}$.

For any fixed $t \in \mathbb{Z}$, $Y(t): \Omega \to \mathbb{R}$ is defined by
$Y(t)(\omega) = f(X(t)(\omega))$. And, for any Borel set $B \subseteq
\mathbb{R}$, its turns out that $Y^{-1}(t) (B) \in F$:
\begin{align*}
Y^{-1} (B) & = \{\omega \in \Omega: Y(t)(\omega) \in B\} \\
	& = \{\omega \in \Omega: f(X(t)(\omega)) \in B\} \\
	& = \{\omega \in \Omega: X(t)(\omega) \in B\} \\
	& = X^{-1} (t)(f^{-1}(B))
  ~.
\end{align*}

Since $f$ is measurable, $f^{-1}(B)$ is a Borel set in $\mathbb{R}$.

Since $X(t)$ is a random variable---as $\{X(t): t \in \mathbb{Z} \}$ is a
stochastic process---then $X^{-1} (t)(f^{-1}(B)) \in F$.

Therefore, $Y^{-1}(t) (B) \in F$ for every Borel set $B$. That is, $Y(t)$ is a
random variable.

Second, the collections of RVs forms a stochastic process.

Since $Y(t)$ is a RV for each $t \in \mathbb{Z}$, the collection $\{Y(t): t \in
\mathbb{Z}\}$ satisfies the definition of a stochastic process. And, its
finite-dimensional distributions over subsets of indices and their
probabilities are well-defined.

That is, the collection forms a stochastic process since $f$'s measurability 
guarantees $f \circ X(t)$ preserves the RV property for each and all $t$.
\end{proof}

\bibliographystyle{unsrtnat}
\bibliography{chaos}

\end{document}